\documentclass[12pt]{article}

\usepackage{amssymb, amsmath, amsthm}
\usepackage{graphicx}
\usepackage{subcaption}
\usepackage[export]{adjustbox}
\usepackage{wrapfig}
\graphicspath{ {./figures/} }
\usepackage{wrapfig}
\usepackage{cite}
\usepackage{enumitem}

\usepackage{xcolor}

\newcommand{\cF}{\mathcal{F}}

\newcommand{\ca}{\mathcal{A}}
\newcommand{\cb}{\mathcal{B}}
\newcommand{\cc}{\mathcal{C}}
\newcommand{\cg}{\mathcal{G}}
\newcommand{\cf}{\mathcal{F}}
\newcommand{\ck}{\mathcal{K}}

\newcommand{\fil}{\zeta} 
\newcommand{\wrpp}{\Omega}
\newcommand{\sgm}{\Sigma}
\newcommand{\nrm}{N}
\newcommand{\aff}{A}
\newcommand{\bff}{B}
\newcommand{\hmu}{\tilde{\mu}}
\newcommand{\scmb}{S}
\newcommand{\knh}{k}

\newtheorem{theorem}{Theorem}
\newtheorem{lemma}{Lemma}
\newtheorem{definition}{Definition}
\newtheorem{cor}{Corollary}
\newtheorem{prop}{Proposition}

\def\be{\begin{equation}}
\def\ee{\end{equation}}
\def\bea{\begin{eqnarray}}
\def\eea{\end{eqnarray}}

\newcommand{\td}{\mathrm{d}}
\usepackage{hyperref}
\usepackage[left=2cm,right=2cm,top=2cm,bottom=2cm]{geometry}

\title{ \bf{All separable supersymmetric AdS$_5$ black holes}}

\author{James Lucietti$^a$\footnote{j.lucietti@ed.ac.uk}\;, Praxitelis Ntokos$^a$\footnote{Praxitelis.Ntokos@ed.ac.uk} and Sergei G. Ovchinnikov$^{a,b}$\footnote{s.g.ovchinnikov@sms.ed.ac.uk} 
\\ \\ \small \sl $^a$School of Mathematics and Maxwell Institute for Mathematical Sciences, \\ \small \sl    University of Edinburgh, King's Buildings, Edinburgh, EH9 3FD, UK \\ \\ \small \sl $^b$Institute of Theoretical and Mathematical Physics, \\ \small \sl  Lomonosov Moscow State University, Moscow 119991, Russia }

\date{}

\begin{document}

\maketitle

\begin{picture}(0,0)(0,0)
\put(350, 240){}
\put(350, 225){}
\end{picture}

\begin{abstract} 
We consider the classification of supersymmetric black hole solutions to five-dimensional STU gauged supergravity that admit torus symmetry. This reduces to a problem in toric K\"ahler geometry on the base space. We introduce the class of separable toric K\"ahler surfaces that unify product-toric, Calabi-toric and orthotoric K\"ahler surfaces, together with an associated class of separable 2-forms.  We prove that  any  supersymmetric toric solution that is timelike, with a separable K\"ahler base space and Maxwell fields, outside a horizon with a compact (locally) spherical cross-section,  must  be   locally isometric to the known black hole or its near-horizon geometry. An essential part of the proof is a near-horizon analysis which shows that the only possible separable K\"ahler base space is Calabi-toric. In particular, this also implies that our previous black hole uniqueness theorem for minimal gauged supergravity applies to the larger class of separable K\"ahler base spaces.

\end{abstract}
\vskip1cm

\newpage

\tableofcontents

\newpage

\section{Introduction}

The AdS/CFT duality predicts that  asymptotically AdS$_5 \times S^5$ supersymmetric black holes in  type IIB supergravity correspond to a class of BPS states in maximally supersymmetric Yang-Mills theory on $\mathbb{R}\times S^3$~\cite{Maldacena:1997re}. A notable problem in this context is to derive the Bekenstein-Hawking entropy of the black holes from the dual Yang-Mills theory~\cite{Kinney:2005ej}.  In recent years, remarkable progress in this area  has led to such a holographic derivation of the entropy of the known black hole~\cite{Hosseini:2017mds, Cabo-Bizet:2018ehj, Choi:2018hmj, Benini:2018ywd, Zaffaroni:2019dhb, Ntokos:2021duk, David:2023gee}. However, a full resolution of this problem naturally requires a complete classification of black holes in this context, which itself is a difficult open problem as explained in~\cite{Lucietti:2021bbh}.

In fact, all known such black hole solutions are solutions to five-dimensional STU gauged supergravity, that is, minimal gauged supergravity coupled to two vector multiplets. This theory arises as a consistent dimensional reduction of type IIB supergravity on $S^5$ that retains only the KK zero modes on the sphere~\cite{Cvetic:1999xp}. The bosonic field content consists of a metric, three Maxwell fields and two real scalar fields. The theory has AdS$_5$ with vanishing matter fields  as the unique maximally supersymmetric solution~\cite{Gutowski:2004yv}. Asymptotically AdS$_5$ solutions in this theory can carry a number of conserved charges: the mass $M$, two angular momenta $J_1, J_2$ and three electric charges $Q_1, Q_2, Q_3$.  Supersymmetric AdS$_5$ solutions satisfy the BPS equality
\be
M= \frac{|J_1|}{\ell}+\frac{|J_2|}{\ell} + |Q_1|+|Q_2|+|Q_3|  \; , \label{BPS}
\ee
where $\ell$ is the AdS$_5$ radius. The most general {\it known} supersymmetric black hole solution in this theory  is a four parameter family, which carries angular momenta $J_1, J_2$ and charges $Q_1, Q_2, Q_3$ related by one nonlinear constraint, with the mass  determined by \eqref{BPS}~\cite{Kunduri:2006ek}.\footnote{A non-extremal black hole that carries six independent conserved charges $M, J_i, Q_I$ has been found~\cite{Wu:2011gq}.  As far as we are aware, it has not been checked that its BPS limit is equal to the known supersymmetric black hole~\cite{Kunduri:2006ek}.}  The special case with equal angular momenta was  first found by Gutowski and Reall~\cite{Gutowski:2004yv}. This immediately raises the following question: is this  the most general supersymmetric AdS$_5$ black hole?  The purpose of this paper is to address this question within five-dimensional STU gauged supergravity.\footnote{The possibility of hairy supersymmetric black holes in other truncations of supergravity has recently been investigated~\cite{Markeviciute:2018yal, Markeviciute:2018cqs,Dias:2022eyq}.}

There is a further truncation of IIB supergravity to five-dimensional minimal gauged supergravity. This  also corresponds to a consistent truncation of STU supergravity obtained by setting the three Maxwell fields equal and the scalars to zero (so the electric charges are equal).  The most general known supersymmetric black hole solution in this theory is the Chong-Cvetic-Lu-Pope (CCLP) solution~\cite{Chong:2005hr}, which corresponds to the equal charge case of the black hole found in~\cite{Kunduri:2006ek}. The special case with equal angular momenta is the Gutowski-Reall (GR) black hole which was the first example of a supersymmetric black hole in AdS$_5$~\cite{Gutowski:2004ez}. We have previously established two uniqueness theorems for  supersymmetric solutions to minimal gauged supergravity that are timelike outside a regular horizon.  The first result of this kind established uniqueness  of such solutions under the assumption of a compatible $SU(2)$ symmetry, establishing that the GR black hole  or its near-horizon geometry are the only solutions in this class~\cite{Lucietti:2021bbh}. The second result established uniqueness for solutions with a compatible toric symmetry and a Calabi-toric K\"ahler base space~\cite{Lucietti:2022fqj}, showing that the CCLP black hole or its near-horizon geometry are the only solutions in this class.  Both of these results make essential use of the uniqueness of the near-horizon geometry~\cite{Gutowski:2004ez, Kunduri:2006uh}, but no global assumptions on the exterior region of the spacetime.  

In this paper we will generalise these results to STU supergravity.  The main challenge is that, unlike in the minimal theory, the Maxwell fields and scalar fields are not fully determined by the K\"ahler base geometry and therefore these must be solved for simultaneously to the metric. The classification of timelike supersymmetric solutions to five-dimensional STU gauged supergravity is equivalent to a problem defined on a K\"ahler base space that is orthogonal to the supersymmetric Killing field~\cite{Gutowski:2004yv}.  This is analogous to the corresponding classification in minimal gauged supergravity~\cite{Gauntlett:2003fk}, where it has been further shown that  supersymmetry reduces to finding a class of K\"ahler metrics that satisfy a complicated fourth order ODE in the curvature~\cite{Cassani:2015upa}.  In the STU theory the full set of supersymmetric constraints has not been previously written down explicitly. We fill this gap (see Section \ref{sec:time-class}) and find that in this more general theory, supersymmetry does not imply an explicit constraint K\"ahler base geometry which is instead coupled to the scalars and the Maxwell fields.  In any case, for solutions that are also invariant under a toric symmetry this reduces to a complicated PDE problem for a  toric-K\"ahler metric.  To make progress we follow the same strategy as in the minimal theory and make further assumptions on the K\"ahler base space.

Motivated by this, we introduce a class of toric-K\"ahler metrics that we refer to as {\it separable} because they can be described in terms of single-variable functions in an orthogonal  coordinate system.  We will show that these naturally unify product-toric, Calabi-toric and ortho-toric K\"ahler metrics.  In fact, these three classes arose in a general study of K\"ahler surfaces admitting a Hamiltonian 2-form, a concept that was introduced in~\cite{Apostolov2001TheGO}. In particular, any toric-K\"ahler surface admitting a Hamiltonian 2-form must be one of these types~\cite{Legendre}.  Therefore, our definition of separable toric-K\"ahler surfaces provides an alternative approach to the study of K\"ahler surfaces that admit Hamiltonian 2-forms which may be worthy of further investigation.   We also define an associated class of separable 2-forms which are similarly described by single-variable functions in the same orthogonal coordinates (the K\"ahler and Ricci form are two such examples). We  then define timelike supersymmetric toric solutions to be separable if they have a separable K\"ahler base and compatible separable Maxwell fields.

We are now ready to state our main results which are summarised by the following  theorems.

\begin{theorem}\label{stu-theorem}
Any supersymmetric toric solution to five-dimensional  STU gauged supergravity, that is timelike and separable outside a smooth horizon with compact (locally) spherical cross-sections, is locally isometric to the known black hole~\cite{Kunduri:2006ek} or its near-horizon geometry.
\end{theorem}

This is a generalisation of the aforementioned theorem proven in minimal supergravity for  Calabi-toric K\"ahler bases~\cite{Lucietti:2022fqj}. The strategy for the proof is as follows. First, we use the classification of near-horizon geometries~\cite{Kunduri:2007qy} to completely fix the single-variable functions of one of the orthogonal coordinates defined by separability (an angular coordinate). Then supersymmetry reduces to an ODE for the single-variable functions of the other orthogonal coordinate  (a radial coordinate), which can be explicitly solved for under  the relevant black hole boundary conditions. A key step in the proof is the near-horizon analysis, which shows that the only separable toric-K\"ahler base space compatible with a smooth horizon is Calabi-toric, that is, the product-toric and orthotoric bases are not allowed. This therefore also gives a stronger form of the theorem in minimal supergravity that we previously established~\cite{Lucietti:2022fqj},  as follows.

\begin{theorem}\label{minimal-theorem}
Any supersymmetric toric solution to five-dimensional minimal  gauged supergravity, that is timelike with a separable K\"ahler base outside a smooth horizon with compact cross-sections, is locally isometric to the CCLP black hole or its near-horizon geometry.
\end{theorem}

A special case of the separable supersymmetric solutions of Calabi-type correspond to (locally) $SU(2)\times U(1)$ invariant solutions.  Surprisingly, the uniqueness proof for this class is more involved and requires the stronger assumption of analyticity (as in the minimal theory).

\begin{theorem}\label{stu-su2-heorem}
Any supersymmetric solution with a local  $SU(2)\times U(1)$ symmetry  to five-dimensional STU gauged supergravity, that is timelike outside an analytic horizon with compact (locally) spherical cross-sections, is locally isometric to the GR black hole~\cite{Gutowski:2004yv} or its near-horizon geometry.
\end{theorem}

We emphasise that these uniqueness theorems do not make any global assumptions on the exterior spacetime such as topology or asymptotics. Therefore, they also rule out the existence of smooth or analytic solutions that are asymptotically locally AdS$_5$ (other than trivial quotients of the known black hole). In particular, this  implies that the supersymmetric black holes with squashed boundary sphere do not have smooth horizons~\cite{Cassani:2018mlh, Bombini:2019jhp} (in the minimal theory it has been shown the horizons are $C^1$ but not $C^2$~\cite{Lucietti:2021bbh}).  We also emphasise that in the STU theory we need to impose the extra assumption that the cross-section is locally $S^3$ (spherical) because, in contrast to the minimal theory, there exist near-horizon geometries with $S^1\times S^2$ (ring) or $T^3$ (torus) topology~\cite{Kunduri:2007qy}. Unfortunately, our techniques do not apply to these other near-horizon geometries because they possess null supersymmetry. Therefore, our result does not address the existence of black rings in STU gauged supergravity, which remans an interesting open problem.

The organisation of this paper is as follows.  In Section \ref{sec:susytoric} we review supersymmetric solutions to five-dimensional gauged supergravity, define the subclass with a compatible toric symmetry, and derive the associated toric data for the near-horizon geometry. In Section \ref{sec:separable} we define the concept of separable toric-K\"ahler metrics and associated separable 2-forms; we have presented this section in a self-contained way as it may be of independent mathematical interest. In Section \ref{sec:separablesusy} we introduce separable supersymmetric solutions and prove the above black hole uniqueness theorems. In Section \ref{sec:discussion} we close with a Discussion. We also include an Appendix with  some auxiliary results on Hamiltonian 2-forms, and a simplified form of the known black hole~\cite{Kunduri:2006ek}  and its near-horizon geometry.

\section{Supersymmetric solutions with toric symmetry}
\label{sec:susytoric}

In this section we first recall the supergravity theory and the known constraints on timelike supersymmetric solutions.  We then impose toric symmetry and examine the constraints on the toric data arising from the presence of a smooth near-horizon geometry.

\subsection{Timelike supersymmetric solutions}\label{sec:time-class}

The bosonic content of five-dimensional $\mathcal{N}=1$ gauged supergravity coupled to $n-1$  vector multiplets comprises of a spacetime metric $\mathbf{g}$, $n$ abelian gauge fields $\mathbf{A}^{I},\,I=1, \dots, n$, and $n-1$ real scalar fields all defined on a five-dimensional manifold $M$.  We work in the conventions of \cite{Kunduri:2007qy} (see also~\cite{Gutowski:2004yv}).  The scalars can be represented by $n$ real positive scalar fields $X^I$ subject to the constraint  
\begin{equation}\label{eq:sc-constraint}
\frac{1}{6}C_{IJK}X^{I}X^{J}X^{K}=1\,,
\end{equation}
where $C_{IJK}=C_{(IJK)}$ are a set of real constants that obey
\be
C_{IJK}C_{J'(LM}C_{PQ)K'}\delta^{JJ'}\delta^{KK'} = 4\delta_{I(L}C_{MPQ)} \; ,  \label{eq:sym}
\ee
which means the scalars are in a symmetric space.   It is convenient to define
\be
X_{I} =\frac{1}{6}C_{IJK} X^J X^K  \; ,
\ee
so that \eqref{eq:sc-constraint} becomes $X_IX^I=1$.
The action is
\be
S= \frac{1}{16 \pi} \int \left( R_{\mathbf{g}}\star 1- Q_{IJ} \mathbf{F}^I \wedge \star \mathbf{F}^J - Q_{IJ} \td X^I \wedge \star \td X^J - \frac{1}{6} C_{IJK} \mathbf{F}^I\wedge \mathbf{F}^J \wedge \mathbf{A}^K  + 2 \ell^{-2} \mathcal{V} \star 1 \right) \; ,
\ee
where $\mathbf{F}^{I}=\td \mathbf{A}^{I}$ are the Maxwell fields and
\be
Q_{IJ}=\frac{9}{2} X_I X_J - \frac{1}{2} C_{IJK}X^K   \; .   \label{eq:QIJ}
\ee
Equation \eqref{eq:sym} ensures the latter is invertible with inverse
\be
Q^{IJ}= 2 X^I X^J- 6 C^{IJK} X_K  \;,
\ee
where $C^{IJK}:= C_{IJK}$ and it follows that
\be
X^I=\frac{9}{2} C^{IJK} X_J X_K  \; .   \label{eq:Xup}
\ee
The scalar potential is
\be
\mathcal{V}= 27 C^{IJK} \bar{X}_I \bar{X}_J X_K \; ,
\ee
where $\bar{X}_I$ are positive constants.  The unique maximally supersymmetric solution of this theory is AdS$_5$ with radius $\ell$ and vanishing Maxwell fields and constant scalars $X^I= \bar{X}^I$~\cite{Gutowski:2004yv}.  We will be mainly interested in STU gauged supergravity which is given by taking $n=3$, $C_{IJK}=1$ if $(IJK)$ is a permutation of $(123)$ and zero otherwise, and $\bar{X}^I=1$ (so $\bar{X}_I=1/3$). The truncation to minimal supergravity is given by constant scalars $X^I=\bar{X}^I$ and equal gauge fields $\mathbf{A}^I= \bar{X}^I\mathbf{A}$ (which is also a truncation of STU supergravity). We find it convenient to introduce a rescaled set of constants $\zeta_I:= 3 \ell^{-1} \bar{X}_I$.

The general form of supersymmetric solutions was determined in~\cite{Gutowski:2004yv}, following the analysis for minimal gauged supergravity~\cite{Gauntlett:2003fk}. Given a supercovariantly constant spinor one can construct several spinor bilinears: a real scalar $f$, a Killing vector field $V$ and three real 2-forms $J^{(i)}$, $i=1,2,3$. These satisfy 
\be
V^\mu V_\mu =-f^2  \; , \label{eq_Vnorm}
\ee
implying that $V$ is either timelike (in some open region) or globally null.  In this paper we will focus on the timelike class, that is, we assume there is some open region $U\subset M$ where $V$ is strictly timelike.   In the timelike case we can assume that $f>0$ in $U$ and write the metric as
\be\label{metricform}
\mathbf{g} = -f^2( \td t+\omega)^2+ f^{-1} h  \; , 
\ee
where $V=\partial_t$, $h$ is a Riemannian metric on the orthogonal base space $B$, and $\omega$ is a 1-form on $B$ defined by $\iota_V \omega=0$ and $\td \omega=- \td (f^{-2} V)$. The 2-forms $J^{(i)}$ can now be regarded as anti-self dual (ASD) 2-forms on $B$ with respect to the volume form $\text{vol}(h)$ on $B$, where  $(\td t +\omega)\wedge \text{vol}(h) $ is positively oriented in spacetime, that also satisfy the algebra of unit quaternions
\be
J^{(i)a}_{\phantom{(i)a}c} J^{(j)c}_{\phantom{(i)c}b} = - \delta^{ij}\delta^a_{~b}+ \epsilon^{ijk}J^{(k)a}_{\phantom{(k)a}b}  \; .  \label{quaternion}
\ee
Supersymmetry then implies that the base space $(B, h)$ is K\"ahler with a  K\"ahler form $J:=J^{(1)}$, 
\be
\nabla_a J^{(2)}_{bc}= P_a J^{(3)}_{bc}, \qquad \nabla_a J^{(3)}_{bc}= - P_a J^{(2)}_{bc}  \; , \label{eq:dJ23}
\ee
where $\nabla$ is the metric connection of $h$, 
\be
P= \zeta_I A^I  \; ,  \label{eq:P}
\ee
 the Maxwell field takes the form
\be
 \mathbf{F}^{I} =\td\big(f X^{I}(\td t+\omega)\big)+F^{I}   \; ,   \label{maxwellform}
\ee
and  the magnetic field $F^I=\td A^I$ is given by
\begin{equation}\label{FI_eq1}
F^{I}=\Theta^{I}-3 f^{-1} C^{IJK}\zeta_JX_K J\,  ,
\end{equation}
for some self-dual (SD) 2-forms $\Theta^{I}$ on $B$ that obey
\begin{equation}\label{GPlus}
G^{+}=-\frac{3}{2}X_{I}\Theta^{I}\,,
\end{equation}
where
\begin{equation}\label{GPlusMinus}
G^{\pm} :=\frac{1}{2}f(\td\omega\pm\star_{4}\td\omega) \; ,
\end{equation}
which encode the SD and ASD parts of $\td \omega$ and $\star_4$ is the Hodge dual on the base.  Equation (\ref{eq:dJ23}) implies that the Ricci form $\mathcal{R}_{ab} := \tfrac{1}{2} R_{abcd} J^{cd}$  is given by   $\mathcal{R}=\td P$ and hence (\ref{eq:P}), (\ref{FI_eq1})  in particular imply that
\be
\mathcal{R}- \frac{1}{4} R J = \zeta_I \Theta^I  \; ,  \label{ricci}
\ee
and
that the scalar $f$ is determined by
\begin{equation}\label{f-expression}
f=-\frac{12}{R}C^{IJK}\fil_{I}\fil_{J}X_{K}\, ,
\end{equation}
where $R$ is the scalar curvature of the base.

The conditions required by supersymmetry must be supplemented by the equations of motion. In particular, the Bianchi identity $\td \mathbf{F}^I=0$ and the Maxwell equations reduce to 
\be
\td F^I=0  \; ,  \label{eq_FIBianchi}
\ee
and
\begin{equation}\label{Maxwell-1}
\text{d}\star_{4}\td(f^{-1}X_{I})= -\frac{1}{6}C_{IJK} \Theta^{J}\wedge\Theta^{K}+\frac{2}{3}\fil_{I}f^{-1}G^{-}\wedge J + \frac{2}{3} f^{-2}( Q_{IJ}C^{JMN}\zeta_M \zeta_N+ \zeta_I \zeta_JX^J) \star_4 1 \; ,
\end{equation}
respectively.  In order to fully determine $\omega$ we expand its ASD part in the basis $J^{(i)}$ by writing
\begin{equation}\label{GMinus}
G^{-}=-\frac{\ell^{3}f}{48}\lambda_{i}J^{(i)}\,,
\end{equation}
where $\lambda_{i}$ are functions on $B$.  As we show below, the Maxwell equation \eqref{Maxwell-1} gives an expression for $\lambda_1$ in terms of $X^I, f$ and $\Theta^I$, whereas $\lambda_2, \lambda_3$ must satisfy a set of PDEs arising from the integrability (i.e. closure) of the equation
\be
\td \omega= f^{-1}( G^+ + G^-)  \; .  \label{dom}
\ee
Conversely, given a K\"ahler base $(B, h, J)$, scalars $X^I$ and SD two-forms $\Theta^I$ that satisfy the above conditions, one can reconstruct a timelike supersymmetric solution with $f, \omega$ determined by \eqref{f-expression} and (\ref{dom}) where $G^\pm$ are given by \eqref{GPlus} and \eqref{GMinus}.

It is convenient to repackage the scalars into a new set of scalar fields $\Phi^I$ defined by
\begin{equation}
\label{eq_Phidef}
\Phi^{I} :=3f^{-1}C^{IJK}\fil_{J}X_{K}\,.
\end{equation}
The inverse transformation is given by
\be
\label{XPhi_eq}
X_I = \frac{f \ell}{3} \left( \frac{1}{2} \ell^2 \zeta_I \zeta_J \Phi^J  -  G_{IJ} \Phi^J \right)  \; ,
\ee
where $G_{IJ}:=2\bar{Q}_{IJ}$ is defined by \eqref{eq:QIJ} with $\bar{X}^I$ in place of $X^I$ (note for STU supergravity $G_{IJ}= \delta_{IJ}$).   We can recover the function $f$ from the scalars $\Phi^I$ by using $C^{IJK}X_I X_J X_K=2/9$ (which follows from  $X^IX_I=1$ and \eqref{eq:Xup}) which gives
\be
f^{-3}= \frac{\ell^3}{6} C^{IJK} \left( \frac{1}{2} \ell^2 \zeta_I \zeta_P \Phi^P  -  G_{IP} \Phi^P \right)\left( \frac{1}{2} \ell^2 \zeta_J \zeta_Q \Phi^Q  -  G_{JQ} \Phi^Q \right)\left( \frac{1}{2} \ell^2 \zeta_K \zeta_R \Phi^R  -  G_{KR} \Phi^R \right)  \; .  \label{eq:fPhi}
\ee
We also introduce a basis of SD 2-forms $I^{(i)}$ that satisfy the quaternion algebra \eqref{quaternion} and expand
\be
\Theta^I= \Theta^I_i I^{(i)}  \; ,
\ee
for functions $\Theta^I_i$.
In terms of these scalars \eqref{FI_eq1} is simply
\begin{equation}\label{FI}
F^{I}=\Theta^{I}-\Phi^I J\,  ,
\end{equation}
so the Bianchi identity is
\be
\td \Theta^I= \td \Phi^I \wedge J  \; ,   \label{eq:bianchi}
\ee
equation \eqref{f-expression} becomes
\begin{equation}\label{Ricciscalar-1}
R=-4\fil_{I}\Phi^{I}\,, 
\end{equation}
and the Maxwell equation \eqref{Maxwell-1} reads
\be
-\nabla^2  \left(G_{IJ} \Phi^J + \frac{1}{8} \ell^2 \zeta_I R \right)   + \frac{1}{\ell} C_{IJK} (\Phi^J \Phi^K-\Theta^J_i\Theta^K_i)+\frac{\ell^2}{12}\zeta_I \lambda_1=0  \; ,   \label{eq:maxwell-2}
\ee
where we have used \eqref{eq:sym}, \eqref{GMinus} and the quaternion algebra \eqref{quaternion}\footnote{This in particular implies $J^{(i)}\wedge J^{(j)}=-2\delta^{ij} \star_4 1$ and $I^{(i)}\wedge I^{(j)}=2\delta^{ij} \star_4 1$.}. Multiplying by $\bar{X}^I$ we can solve the  equation \eqref{eq:maxwell-2} for $\lambda_1$ (note $\bar{X}^I \zeta_I= 3\ell^{-1}$) resulting in
\be
\lambda_1= \frac{1}{2}\nabla^2 R  + \frac{4}{\ell^2} C_{IJK} \bar{X}^I  (\Theta^J_i\Theta^K_i- \Phi^J \Phi^K ) \; .   \label{eq:lambda1}
\ee
Equation \eqref{dom} now reads
\be
\td \omega = \frac{\ell}{2} \left( G_{IJ} \Phi^J+ \frac{\ell^2}{8} \zeta_I R \right) \Theta^I- \frac{\ell^3}{48} \lambda_i J^{(i)}  \; .  \label{eq:dom-2}
\ee
The integrability condition for this latter equation, namely that the r.h.s. is a closed 2-form reduces to a constraint as follows.

Using the duality properties of $\Theta^I$ and $J^{(i)}$ we can write the integrability condition $\td^2 \omega=0$  in the equivalent form
\be
\nabla^a \left[ \left( G_{IJ} \Phi^J+ \frac{\ell^2}{8} \zeta_I R \right) \Theta_{ab}^I  J^{b}_{\phantom{\,b}c}  \right]  - \frac{\ell^2}{24} \left[ \nabla_c \lambda_1- (\nabla^a\lambda_3+P^a \lambda_2)J^{(2)}_{ac}  + (\nabla^a\lambda_2-P^a \lambda_3)J^{(3)}_{ac} \right]=0  \; , \label{lambda23eq}
\ee
where we have used (\ref{quaternion}) and (\ref{eq:dJ23}). Equation (\ref{lambda23eq}) should be interpreted as a PDE for $(\lambda_2, \lambda_3)$ with $\lambda_1$ fixed by (\ref{eq:lambda1}). The integrability condition for (\ref{lambda23eq})  is given by its divergence which is,
\be 
\nabla^c \nabla^a \left[ \left( G_{IJ} \Phi^J+ \frac{\ell^2}{8} \zeta_I R \right) \Theta_{ab}^I  J^{(1)b}_{\phantom{(1)b}c}  \right]  - \frac{\ell^2}{24}\nabla^2 \lambda_1 =0  \; ,    \label{eq:susyPDE}
\ee
where we have used \eqref{ricci} together with orthogonality of SD and ASD 2-forms.  Observe that the constraint (\ref{eq:susyPDE}) with $\lambda_1$ given by (\ref{eq:lambda1}) involves only the scalars, the SD two-forms $\Theta^I$ and the K\"ahler base geometry.  For minimal supergravity, it can be checked that \eqref{eq:susyPDE} reduces to the following constraint on the curvature of the K\"ahler base~\cite{Cassani:2015upa},
\be
\nabla^2 \left(\frac{1}{2}\nabla^2 R+ \frac{2}{3}R_{ab}R^{ab}- \frac{1}{3}R^2  \right)+ \nabla^c(R_{ca}\nabla^a R)=0\; ,
\ee
where we have used $\mathcal{R}_{ab}= R_{ac} J^c_{~b}$.  Interestingly, in contrast to the minimal theory, in the  general theory considered here, the constraint \eqref{eq:susyPDE} is not purely in terms of the base geometry, and therefore it is unclear what constraints (if any) there are on the K\"ahler base geometry.

We can now summarise the construction of any timelike supersymmetric solution.  Choose a K\"ahler base $(B, h, J)$ and a set of SD two-forms $\Theta^I$ and scalar fields $\Phi^I$ on $B$ that obey \eqref{ricci}, \eqref{Ricciscalar-1}, \eqref{eq:bianchi}, \eqref{eq:maxwell-2},  \eqref{eq:susyPDE}, where $\lambda_1$ is given by \eqref{eq:lambda1}.  Then we can solve \eqref{lambda23eq} for $\lambda_2, \lambda_3$, since \eqref{eq:susyPDE} is the integrability condition for this equation.  Next we can solve \eqref{eq:dom-2} for $\omega$, since \eqref{lambda23eq} is the integrability condition for this equation.  The function $f$ is simply given by \eqref{eq:fPhi}. The  spacetime metric is then given by \eqref{metricform} and the original scalars by \eqref{XPhi_eq}. Finally, the Maxwell field is given by \eqref{FI} and \eqref{maxwellform}.  

The decomposition of the solution in terms of the supersymmetric data described above is defined up to constant rescalings (since the Killing spinor is). These rescale the time coordinate adapted to $V$ as $t\to Kt$, where $K$ is a nonzero constant, and act on the supersymmetric data as 
\begin{equation}\label{time-resc}
h\to K^{-1}h\,,\qquad \Phi^{I}\to K \Phi^{I}\,, 
\qquad \lambda_{i}\to K^{2}\lambda_{i}\,, \qquad \Theta^{I}\to \Theta^{I}\, ,
\end{equation}
which also implies that $\omega\to K\omega$ and $f\to K^{-1}f$. It is easy to check that the above equations, and in particular the five-dimensional solution $(\mathbf{g}, \mathbf{F}^{I}, X^{I})$,  are invariant under such rescalings.

We close this section by noting a particular consequence of supersymmetry.   This result will be useful in constraining the form of the Maxwell field for toric solutions.
 \begin{lemma}\label{J-inv-Maxwell}
The magnetic part of the Maxwell fields $F^I$  are $J$-invariant, that is, they obey
\be
 F^I_{cd} J^c_{~a} J^d_{~b} = F^I_{ab}   \label{eq_FIJinv} \; .
\ee
\end{lemma}
\begin{proof}
Consider the map on 2-forms $\mathcal{J}:\Omega^2(B)\to \Omega^2(B)$ defined by $(\mathcal{J}\alpha)_{ab}:= \alpha_{cd} J^c_{~a} J^d_{~b}$.  It is easy to see that $\mathcal{J}^2=\text{id}$ and hence $\Omega^2(B)$ decomposes into two eigenspaces of $\mathcal{J}=1$ eigen-2-forms ($J$-invariant) and $\mathcal{J}=-1$ eigen-2-forms.  One can check that the $\mathcal{J}=1$ eigenspace is spanned by the three SD 2-forms $I^{(i)}$ and the ASD  K\"ahler form $J=J^{(1)}$, whereas the $\mathcal{J}=-1$ eigenspace is spanned by remaning ASD forms $J^{(2)}, J^{(3)}$.   

The lemma now follows from equation (\ref{FI}), which show that the magnetic part of the Maxwell field is in the $\mathcal{J}=1$ eigenspace.  
\end{proof}

\subsection{Toric symmetry}\label{sec:toric}

We  now consider supersymmetric solutions $(M, \mathbf{g}, \mathbf{F}^I, X^I)$ to five-dimensional  gauged supergravity as above, that also possess toric symmetry in the following sense.

\begin{definition} A supersymmetric solution $(M, \mathbf{g}, \mathbf{F}^I, X^I)$ to five-dimensional minimal gauged supergravity coupled to vector multiplets is said to possess toric symmetry if:
\begin{enumerate}
\item There is a torus $T\cong U(1)^2$ isometry generated by spacelike Killing fields $m_i$, $i=1,2$, both normalised to have $2\pi$ periodic orbits. These are defined up to $m_i \to A_{i}^{~j} m_j$ where $A\in GL(2, \mathbb{Z})$;
\item  The supersymmetric Killing $V$ is complete and commutes with the $T$-symmetry, that is $[V, m_i]=0$, so there is a spacetime isometry group  $\mathbb{R}\times U(1)^2$;
\item 
The Maxwell fields and scalar fields are $T$-invariant $\mathcal{L}_{m_i} \mathbf{F}^I=0$, $\mathcal{L}_{m_i} X^I=0$; 
\item The axis defined by  $\{p\in M | \det \mathbf{g}(m_i, m_j)|_{p}=0 \}$
is nonempty.
\end{enumerate}
\end{definition}

We now restrict to timelike supersymmetric solutions, that is, we assume there is an open region $U\subset M$ on which $V$ is strictly timelike. We will assume that $U$ is simply connected and intersects the axis.  Therefore, on $U$, one can write the metric as \eqref{metricform}, where \eqref{eq_Vnorm} holds and the K\"ahler metric on the orthogonal base space $B$ is $h_{\mu\nu}=g_{\mu\nu}+ V_\mu V_{\nu}/f^2$. It follows from the above assumptions that the data $f, h$ on the base are invariant under the toric symmetry.  Thus, it also follows that the scalar fields $\Phi^I$ defined by \eqref{eq_Phidef} are  invariant under the toric symmetry. The 1-form $\omega$ is defined up to gauge transformations $\omega\to \omega+ \td \lambda$ and $t\to t-\lambda$, where $\lambda$ is a function on $B$, so one can choose  a gauge such that $\mathcal{L}_{m_i}\omega=0$ and $\mathcal{L}_{m_i} t=0$.  The form of the Maxwell field (\ref{maxwellform}), (\ref{FI}), (\ref{GPlus}) implies
\begin{equation}
X_I \mathbf{F}^I = X_I \td ( f X^I (\td t +\omega))- \frac{2}{3} G^+-  X_I \Phi^I J  \; ,
\end{equation}
so we deduce that the K\"ahler form is also invariant under the toric symmetry $\mathcal{L}_{m_i} J=0$, i.e.  the toric symmetry is holomorphic. It follows that $\iota_{m_i} J$ is a globally defined closed 1-form so must be exact on $B$. This shows that the toric symmetry is Hamiltonian  and  hence $(B, h, J)$ is a K\"ahler toric manifold. 

It is convenient to introduce symplectic coordinates $(x_i, \phi^i)$ on $B$, $i=1,2$, such that $m_i=\partial_{\phi^i}$~\cite{Abreu, Lucietti:2022fqj},
\bea
&&h = G^{ij}(x) \td x_{i} \td x_{j}+ G_{ij}(x) \td\phi^i \td\phi^j, \label{h_toric}  \\
&&G^{ij} = \partial^{i}\partial^{j} g  \; ,\label{hessiang} \\
&& J= \td x_{i} \wedge \td\phi^i , \label{KahlerForm-gen}
\eea
where $g=g(x)$ is the symplectic potential, $G_{ij}$ is the matrix inverse of the Hessian $G^{ij}$ and we  have introduced the notation $\partial^{i}:=\partial/\partial{x_i}$.    In these coordinates the Ricci form potential is 
\begin{equation}
P=P_{i}\td\phi^{i}\,,\qquad P_{i}=-\frac{1}{2}G_{ij}\partial^{j}\log\det G=-\frac{1}{2}\partial^{j}G_{ij}\,,    \label{eq_ricci_toric}
\end{equation}
where $\det G := \det G_{ij}$.  It is useful to note that symplectic coordinates are related to holomorphic coordinates $z^i:=y^i+i \phi^i$ by the transformation 
\be
y^{i}=\partial^{i}g \; ,
\ee
and the Legendre transform of the symplectic potential,
\begin{equation}\label{Kahler-Legendre}
\ck =x_{i}y^{i}-g\,  ,
\end{equation}
gives the K\"ahler potential $\mathcal{K}$~\cite{Abreu} (see also appendix A of \cite{Lucietti:2022fqj}).

Let us now turn our attention to closed 2-forms on toric K\"ahler manifolds. If such a closed 2-form $F$ is invariant under the torus symmetry, then $\iota_{m_i} F$ is a closed 1-form on $B$ and hence must be exact, so $\iota_{m_i} F= -\td \mu_i$ for some functions $\mu_i$ (moment maps for $F$). Furthermore, we must also have that $\iota_{m_1} \iota_{m_2} F$ is a constant, and since we assume the axis is nonempty this constant  is zero so  $\iota_{m_1} \iota_{m_2} F=0$. Therefore, the functions $\mu_i$ are also invariant under the $T$-symmetry. Observe that in general, $F$ is not uniquely determined by its moment maps. To this end, it is convenient to introduce the following class of 2-forms.

\begin{definition}\label{def-toric-2forms}
A  closed 2-form $F$ on a toric K\"ahler manifold is said to be toric if it is invariant under the toric symmetry and satisfies the orthogonality condition $m_1^\flat\wedge m_2^\flat \wedge F=0$, where $m_i^\flat$ are 1-forms that are metric dual to the Killing fields $m_i$.
\end{definition}

It is straightforward to show that in symplectic coordinates \eqref{h_toric} the condition $m_1^\flat\wedge m_2^\flat \wedge F=0$ is equivalent to $F_{x_i x_j}=0$, and hence a toric closed 2-form $F$ takes the general form
\be  \label{mu-i-introduction}
F= \td (\mu_i \td \phi^i)  \; ,
\ee
so in particular it is  determined by the moment maps $\mu_i$. 

An example of a toric closed 2-form is $\td\omega$ as we now show. First observe that the magnetic fields $F^I$  are $T$-invariant, since the Maxwell fields \eqref{maxwellform} are, and satisfy the Bianchi identity (\ref{eq_FIBianchi}), hence we deduce $\iota_{m_1}\iota_{m_2} F^I=0$. Since we also have $\iota_{m_1}\iota_{m_2} J=0$, (\ref{FI}) implies $\iota_{m_1} \iota_{m_2} \Theta^I=0$ and in turn (\ref{GPlus}) implies $\iota_{m_1}\iota_{m_2} G^+=\iota_{m_1}\iota_{m_2}\ast_{4}\td\omega=0$~\footnote{Recall that $\td\omega$ is closed and $T$-invariant and hence $\iota_{m_1}\iota_{m_2}\td\omega=0$.}. The latter condition is equivalent to $m_1^\flat\wedge m_2^\flat \wedge \td\omega=0$ which shows that $\td\omega$ is toric and therefore can be written as in \eqref{mu-i-introduction}, so
\begin{equation}
\label{om-toric}
\omega= \omega_i(x) \td \phi^i  \; ,
\end{equation}
for functions $\omega_i$  that are invariant under the toric symmetry.

We now establish the following useful result for a class of 2-forms on the K\"ahler base.
\begin{lemma} \label{lem_JinvF_toric}
Let $F \in \Omega^2(B) $ be closed, $J$-invariant, and invariant under the toric symmetry.  Then $F$ is toric in the sense of Definition \ref{def-toric-2forms}, and in symplectic coordinates
\be 
\partial^i (G^{jk}\mu_k)-\partial^j (G^{ik}\mu_k)=0  \; ,
\label{eq_JinvF}
\ee
where $\mu_{i}$ are the $T$-invariant moment-maps of $F$. Thus we may write
\be
\mu_i= G_{ij}\partial^j \Lambda  \; ,  \label{eq_Lambda}
\ee
where $\Lambda$ is a $T$-invariant function.
\end{lemma}

\begin{proof}
 By the comments preceding Definition \ref{def-toric-2forms}, in symplectic coordinates $0=\iota_{m_i}\iota_{m_j}F=F_{\phi^i\phi^j}$.  Next, $J$-invariance means $J_{a}^{~c}J_{b}^{~d}F_{cd}= F_{ab}$, which reduces to $F_{x_i x_j}=0$ and (\ref{eq_JinvF}).   The latter equation just says that the 1-form $G^{ij}\mu_j \td x_i$ is closed and hence must equal $\td \Lambda$ for some function $\Lambda$ invariant under the toric symmetry.
\end{proof}

  It is worth noting that in the holomorphic coordinates $y^i+i \phi^i$  equations \eqref{eq_JinvF} and \eqref{eq_Lambda} are simply $\partial_{[i}\mu_{j]}=0$ and $\mu_i= \partial_i \Lambda$ respectively where $\partial_i:= \partial/\partial y^i$. We also note that the moment maps $\mu_i$ are defined up to gauge transformations,
  \begin{equation}\label{U1-gauge}
\mu_{i}\to\mu_{i}+\alpha_{i}\, ,
\end{equation}
where $\alpha_i$ are constants.
Lemma \ref{lem_JinvF_toric} in particular applies to geometric quantities such as the K\"ahler form $J$ and the Ricci form $\mathcal{R}$, since they are both closed $J$-invariant 2-forms invariant under the toric symmetry.  The associated potentials (\ref{eq_Lambda})  for the K\"ahler form can be read off from $x_i= G_{ij}\partial^j \mathcal{K}$ (this follows from (\ref{Kahler-Legendre})), whereas for the Ricci form it can be read off from (\ref{eq_ricci_toric}), and are
\begin{equation}
\Lambda_{\text{K\"ahler-form}}= \mathcal{K}, \qquad \Lambda_{\text{Ricci-form}}=-\frac{1}{2}\log\det G\, .
\end{equation}
Lemma \ref{J-inv-Maxwell} shows that Lemma  \ref{lem_JinvF_toric} also applies to the magnetic fields, so  they can be written as
 \begin{equation}\label{}
 F^I  = \td (\mu_i^I \td \phi^i) \; ,  \qquad \mu_i^I= G_{ij}\partial^j \Lambda^I \; ,
 \end{equation}
 where the `magnetic' potentials $\Lambda^I$ are invariant under the toric symmetry.
  
 To summarise, we have shown that for a timelike supersymmetric toric solution we can write the spacetime metric and gauge fields in symplectic coordinates as 
\begin{align}\label{toric-generic}
\mathbf{g} & =-f^{2}(\td t+\omega_{i}\td\phi^{i})^{2}+f^{-1}G_{ij}\td\phi^{i}\td\phi^{j}+G^{ij}\td x_{i}\td x_{j}\,,\nonumber \\
\mathbf{A}^{I} & =fX^{I}(\td t+\omega_{i}\td\phi^{i})+\mu_{i}^{I}\td\phi^{i}\,.
\end{align}
We emphasise that any such solution is determined by the following $T$-invariant real functions,  the symplectic potential $g$, the magnetic potentials $\Lambda^I$, and the scalar fields $X^I$ (or $\Phi^I$), subject to the constraints presented in Section \ref{sec:time-class}.
It is useful to record the following spacetime invariants
\begin{align}
\mathbf{g}(V,V) &=-f^{2}\,,\qquad\mathbf{g}(V,m_{i})=-f^{2}\omega_{i}\,,\qquad\mathbf{g}(m_{i},m_{j})=f^{-1}G_{ij}-f^{2}\omega_{i}\omega_{j}\,, \nonumber\\
\iota_V \mathbf{A}^{I} &=fX^{I}\,,\qquad \iota_{m_i}\mathbf{A}^{I}=fX^{I}\omega_{i}+\mu_{i}^{I}\, ,  \label{eq_invariants}
\end{align}
and observe that these are all invariant under the toric symmetry.

The axis of the $T$-symmetry has a simple description in symplectic coordinates~\cite{Lucietti:2022fqj}.  In particular, each component corresponds to a line segment $\ell_v(x):= v^ix_i+c_v = \text{const}$, where its slope $v^i$ corresponds to the vector $v^i\partial_\phi^i$ that vanishes on the particular axis component. The singular behaviour of the symplectic potential near any component of the axis then takes a canonical form $g =  \frac{1}{2} \ell_v(x) \log \ell_v(x) +g_{\text{smooth}}$ where the latter term is smooth at the said axis. This analysis does not depend on the matter content, for more details see~\cite{Lucietti:2022fqj}.

\subsection{Near-horizon geometry}\label{sec:NH}

We are interested in solutions that possess black hole regions.  The event horizon of a black hole must be invariant under any Killing field and hence in particular under the supersymmetric Killing field.  Thus the horizon is a supersymmetric horizon and hence must be a degenerate Killing horizon with respect to $V$~\cite{Kunduri:2013gce}.  We will assume that a connected component of the  horizon has a spacelike cross-section $S$ transversal to $V$. Then the metric near this horizon component can be written in Gaussian null coordinates (GNC)~\cite{Moncrief:1983xua}, $(v, \lambda, y^a)$, (see also~\cite{Kunduri:2007qy}),
\begin{equation}\label{NH-metric-main}
\mathbf{g}=-\lambda^{2}\Delta^{2}\td v^{2}+2\td v\td\lambda+2\lambda h_{a}\td v\td y^{a}+\gamma_{ab}\td y^{a}\td y^{b}\,,
\end{equation}
where $V=\partial_v$,  $\partial_\lambda$ is a transverse null geodesic field synchronised so $\lambda=0$ at the horizon, and $(y^a)$ are coordinates on $S$.  We assume the horizon is smooth which means that the metric data $\Delta, h_a, \gamma_{ab}$  are smooth functions of $(\lambda, y^a)$ at $\lambda=0$. The near-horizon geometry is given by replacing this data in the spacetime metric with their values at $\lambda=0$, denoted by $(\Delta^{(0)}, h^{(0)}_a, \gamma_{ab}^{(0)})$, which correspond to a function, 1-form and Riemannian metric on $S$ respectively. The gauge fields in GNC can be written in the gauge~\cite{Kunduri:2007qy},
\begin{equation}\label{NH-gauge-main}
\mathbf{A}^{I}=\lambda \Delta X^{I}\td v + A^I_\lambda \td \lambda +a_{a}^{I}\td y^{a}\,,
\end{equation}
where $A_\lambda^I$ does not appear in the near-horizon limit, so the near-horizon data of the  gauge fields is given by the  functions $\Delta^{(0)} X^{I (0)}$ and 1-forms $a^{I(0)}$ on $S$ where  $(0)$ denotes again evaluation at $\lambda=0$. The toric Killing fields $m_i$ must be tangent to the horizon and since they have closed orbits we may choose them tangent to $S$ (hence one can adapt coordinates $(y^a)$ on $S$ to the toric Killing fields).

A horizon component corresponds to a single point in symplectic coordinates $(x_1, x_2)$ on the orbit space. The proof of this does not require detailed knowledge of the near-horizon geometry and is identical to that in the minimal theory, see~\cite[Lemma 3]{Lucietti:2022fqj}.  Its proof uses  the fact that the K\"ahler form in GNC is
\begin{equation}
\label{eq_JGNC}
J=\td\lambda\wedge Z+\lambda(h\wedge Z-\Delta\star_{3}Z)\,,
\end{equation}
where $Z=Z_a \td y^a$ is a unit 1-form and $\star_3$ is the Hodge star operator with respect to $\gamma_{ab}$. The symplectic coordinates are determined by $\td x_i=  -\iota_{m_i} J$, so using the above form for $J$ gives
\begin{equation}\label{xiZ0}
x_{i}=\lambda Z_{i}+O(\lambda^{2})\,, 
\end{equation}
where we have assumed that $m_i$ are tangent to $S$, $Z_i :=\iota_{m_i} Z$, and fixed an integration constant. Thus the horizon corresponds to a point in symplectic coordinates on the orbit space, which we have assumed to be at the origin.

We now turn to the explicit form of the near-horizon geometry.  The classification of near-horizon geometry admitting a torus rotational isometry that commutes with $V$, and possessing a smooth compact  cross-section $S$, was derived in~\cite{Kunduri:2007qy}.   We will only consider the case where $S$ is topologically $S^3$ or a quotient, which includes the possibility of lens spaces. We present it here in a coordinate system that also describes the special case with (local) $SU(2)\times U(1)$ symmetry, see Appendix \ref{NH-comparison} for details.  For simplicity, henceforth we will restrict ourselves to the STU supergravity. 

The near-horizon geometry depends on the parameters $0< \ca^{2},\cb^{2}<1$ and $\ck_I$, subject to the constraints
\begin{equation}\label{kappa2}
\kappa^{2}(\ca^{2},\cb^{2},\cc_{1},\cc_{2})>0\,,
\end{equation}
where
\begin{align}\label{kappa-def}
\kappa^{2} (\ca^{2},\cb^{2},\cc_{1},\cc_{2}) & :=-9\mathcal{A}^{4}\mathcal{B}^{4}+6\mathcal{A}^{2}\mathcal{B}^{2}(\mathcal{A}^{2}+\mathcal{B}^{2}+1)^{2}-(\mathcal{A}^{2}+\mathcal{B}^{2}+1)^{3}\Big(\mathcal{A}^{2}+\mathcal{B}^{2}-\frac{1}{3}\Big)\nonumber \\
 & +\frac{4\cc_{2}}{3}-2\cc_{1}\Big(\mathcal{A}^{4}+\mathcal{B}^{4}-\mathcal{A}^{2}\mathcal{B}^{2}+\frac{\cc_{1}}{2}-1\Big)\, ,
\end{align}
 the parameters $\cc_{1},\cc_{2}$ are defined by
\begin{equation}\label{ccal-def}
\cc_{1}=\frac{\ell}{6}C^{IJK}\fil_{I}\ck_{J}\ck_{K}\,,\qquad\cc_{2}=\frac{1}{6}C^{IJK}\ck_{I}\ck_{J}\ck_{K}\, ,
\end{equation}
and
\begin{equation}\label{ck-constraint}
C^{IJK}\fil_{I}\fil_{J}\ck_{K}=0\,.
\end{equation}
We now give the explicit expression for the near-horizon data. 

The metric data are given by
\begin{align}\label{NH-expl1}
\Delta^{(0)} & =\frac{3\kappa}{\ell\hat{H}(\hat{\eta})^{2/3}}\,, \nonumber \\
h^{(0)} & =\frac{3\kappa\Delta_{3}(\hat{\eta})}{4\hat{H}(\hat{\eta})}\hat{\sigma}+\frac{3(\ca^{2}-\cb^{2})}{2\hat{H}(\hat{\eta})}\Big(\big(\Delta_{2}(\hat{\eta})^{2}+\cc_{1}\big)\td\hat{\eta}-\frac{3}{2}\kappa\hat{\tau}\Big)\,, \nonumber \\
\gamma^{(0)} & =\frac{\ell^{2}}{12(1-\hat{\eta}^{2})\Delta_{1}(\hat{\eta})}\Big(\hat{H}(\hat{\eta})^{1/3}\td\hat{\eta}^{2}+\frac{3}{4}\frac{\Delta_{3}(\hat{\eta})^{2}+\kappa^{2}}{\hat{H}(\hat{\eta})^{2/3}}\hat{\tau}^{2}\Big)\nonumber \\
 & +\frac{\ell^{2}\big(4\hat{H}(\hat{\eta})-3\Delta_{3}(\hat{\eta})^{2}\big)}{48\hat{H}(\hat{\eta})^{2/3}}\hat{\sigma}^{2}+\frac{3\ell^{2}\Delta_{3}(\hat{\eta})(\ca^{2}-\cb^{2})}{8\hat{H}(\hat{\eta})^{2/3}}\hat{\sigma}\hat{\tau}\,,
\end{align}
the gauge field data by
\begin{align}\label{NH-expl2}
a_{(0)}^{I} & =-\frac{\ell\Delta_{3}(\hat{\eta})}{4(\ell\Delta_{2}(\hat{\eta})\fil_{I}+\ck_{I})}\hat{\sigma}+\ell\hat{\eta}\frac{C^{IJK}\ell\fil_{J}\big(\ell(2-\ca^{2}-\cb^{2})\fil_{K}+2\ck_{K}\big)}{12(1-\hat{\eta}^{2})\Delta_{1}(\hat{\eta})}\hat{\tau}\nonumber \\
 & +\ell\frac{(\ca^{2}-\cb^{2})}{4}\Big(\frac{3}{\ell\Delta_{2}(\hat{\eta})\fil_{I}+\ck_{I}}+\ell^{2}\frac{C^{IJK}\fil_{J}\fil_{K}}{2\Delta_{1}(\hat{\eta})}\Big)\hat{\tau}\,,
\end{align}
and the scalars by
\begin{equation}\label{NH-expl3}
X_{I}^{(0)}=\frac{\ell\Delta_{2}(\hat{\eta})\zeta_I +\ck_{I}}{3\hat{H}(\hat{\eta})^{1/3}}\,.
\end{equation}
In the above expressions we have defined the 1-forms
\begin{equation}
\hat{\sigma}=\frac{1-\hat{\eta}}{\ca^{2}}\td\hat{\phi}^{1}+\frac{1+\hat{\eta}}{\cb^{2}}\td\hat{\phi}^{2}\,,\qquad\hat{\tau}=(1-\hat{\eta}^{2})\Delta_{1}(\hat{\eta})\Big(\frac{\td\hat{\phi}^{1}}{\ca^{2}}-\frac{\td\hat{\phi}^{2}}{\cb^{2}}\Big)\,,
\end{equation}
three linear functions of $\hat{\eta}$,
\begin{align}\label{Delta-functions}
\Delta_{1}(\hat{\eta}) & =\frac{1+\hat{\eta}}{2}\ca^{2}+\frac{1-\hat{\eta}}{2}\cb^{2}\,,\nonumber \\
\Delta_{2}(\hat{\eta}) & =1-\frac{1+3\hat{\eta}}{2}\ca^{2}-\frac{1-3\hat{\eta}}{2}\cb^{2}\,,\nonumber \\
\Delta_{3}(\hat{\eta}) & =1-2\Delta_{2}(\hat{\eta})+\ca^{2}\cb^{2}-\ca^{4}-\cb^{4}-\cc_{1}\,,
\end{align}
and the cubic polynomial of $\hat{\eta}$,
\begin{equation}
\label{eq:H-function}
\hat{H}(\hat{\eta})=\prod_{I=1}^{3}\big(\ell \Delta_{2}(\hat{\eta})\zeta_I+\ck_{I}\big)=\Delta_{2}(\hat{\eta})^{3}+3\cc_{1}\Delta_{2}(\hat{\eta})+\cc_{2}\,.
\end{equation}
Here  $(\hat{\eta}, \hat{\phi}^{i})$ are coordinates on $S$ with $-1\leq\hat{\eta}\leq1$ and ${ \hat{\phi}^{i}}\sim{ \hat{\phi}^{i}}+2\pi$ are adapted to the Killing fields $m_{i}=\partial_{ \hat{\phi}^{i}}$ and  $\Delta_1$ and $\Delta_2$ are strictly positive functions. The 1-form that determines the K\"ahler form \eqref{eq_JGNC} is given by
\begin{equation}\label{Z0}
Z^{(0)}=\frac{\ell}{4\hat{H}(\hat{\eta})^{1/3}}\big(\kappa\hat{\sigma}-3(\ca^{2}-\cb^{2})\td\hat{\eta}\big)\,.
\end{equation}
Note that solutions with $\ca^{2}\neq\cb^{2}$ are doubly counted with the two copies related by \eqref{exchange}.  The solutions with $\ca^2=\cb^2$ have enhanced (local) $SU(2)\times U(1)$ symmetry.  

 It is important to emphasise that  positivity of the scalars $X^I$ places further constraints on the parameters as follows.  Without loss of generality we may  assume $\mathcal{A}^2 \geq \mathcal{B}^2$ (which removes the double counting mentioned above), in which case positivity of the scalars is equivalent to
\begin{equation}\label{c1c2-bounds}
\mathcal{K}_{I}>2\ca^{2}-\cb^{2}-1\,,
\end{equation}
where note that $2\ca^{2}-\cb^{2}-1<0$. It is useful to note that (\ref{c1c2-bounds}) implies that 
\begin{equation}
-(1+\cb^{2}-2\ca^{2})^{2}<\mathcal{C}_{1}\leq0\,,\qquad-\frac{1}{4}(1+\cb^{2}-2\ca^{2})^{3}<\mathcal{C}_{2}<2(1+\cb^{2}-2\ca^{2})^{3}\,, 
\end{equation}
which again holds assuming $\mathcal{A}^2 \geq \mathcal{B}^2$.

We are now ready to extract the toric data for the near-horizon geometry.  By computing the spacetime invariants (\ref{eq_invariants}) in GNC we deduce 
\begin{align}
f &=\lambda\Delta\,,\qquad\omega_{i}=-\frac{h_{i}}{\lambda\Delta^{2}}\,,\qquad G_{ij}=\lambda\Delta\Big(\gamma_{ij}+\frac{h_{i}h_{j}}{\Delta^{2}}\Big)\,, \\ 
\mu_{i}^{I} &=a_{i}^{I}+\frac{f X^{I}h_{i}}{\Delta^{2}}\, .   \label{mu-NH}
\end{align}
We may now compute the near-horizon behaviour of $\omega_i, G_{ij}, X^I$ from the explicit near-horizon data and we find
\begin{align}
\omega_{i} &=-\frac{\ell^{2}\hat{H}(\hat{\eta})^{1/3}}{12\kappa}\Big(\Delta_{3}(\hat{\eta})\hat{\sigma}-3(\ca^{2}-\cb^{2})\hat{\tau}\Big)_{i}\,\frac{1}{\lambda}+O(1)\,,  \label{omNH} \\
\label{Gdd-NH}
G_{ij} &=\frac{\ell\kappa}{4\hat{H}(\hat{\eta})^{1/3}}\Big(\hat{\sigma}^{2}+\frac{\hat{\tau}^{2}}{(1-\hat{\eta}^{2})\Delta_{1}(\hat{\eta})}\Big)_{ij}\,\lambda+O(\lambda^{2})\, , 
\\
X^{I} &=\frac{\hat{H}(\hat{\eta})^{1/3}}{\ell\Delta_{2}(\hat{\eta})\zeta_I+\ck_{I}}+O(\lambda)\,.
\end{align}
The near-horizon behaviour of the scalars $\Phi^I$ defined by \eqref{eq_Phidef} is easily deduced to be
\begin{equation}\label{Phi-NH}
\Phi^{I}=\frac{\ell \hat H(\hat\eta)^{1/3}}{3\kappa} C^{IJK} \zeta_J (\ell\Delta_2(\hat\eta)\zeta_K +\mathcal{K}_K) \lambda^{-1}+O(1)\, .
\end{equation}
Observe that although $\Phi^I$ are singular on the horizon the scalars $X^I$ defining the theory are smooth.

Inserting \eqref{Z0} into \eqref{xiZ0} we find the leading order coordinate change is
\begin{equation}\label{xi-NH-expl}
x_{1}=\frac{\ell\kappa}{4\hat{H}(\hat{\eta})^{1/3}}\frac{1-\hat{\eta}}{\ca^{2}}\lambda+O(\lambda^{2})\,,\qquad x_{2}=\frac{\ell\kappa}{4\hat{H}(\hat{\eta})^{1/3}}\frac{1+\hat{\eta}}{\cb^{2}}\lambda+O(\lambda^{2})\,.
\end{equation}
One can now compute $G_{ij}$,  $G^{ij}$ and hence $g$ as functions of the symplectic coordinates near the horizon.  We find that the symplectic potential takes precisely the same form as in the minimal theory, namely~\cite{Lucietti:2022fqj}, 
\be
g= \frac{1}{2} x_1\log x_1+ \frac{1}{2} x_2\log x_2-\frac{1}{2} (x_1+x_2) \log (x_1+x_2)+ \frac{1}{2} (\mathcal{A}^2x_1+\mathcal{B}^2 x_2) \log(\mathcal{A}^2x_1+\mathcal{B}^2 x_2)  + \tilde{g}
\ee
where $\tilde{g}$ is smooth at the origin (horizon). 

We have now determined the behaviour of the symplectic potential near any component of a horizon. Combining this with the near axis behaviour discussed earlier, we can write down the singular part of the symplectic potential for any black hole solution in this class. The result is the same as in the minimal theory~\cite[Theorem 2]{Lucietti:2022fqj}.

\section{Separability on toric K\"ahler manifolds}\label{sec:separable}

In this section we will introduce a special class of toric K\"ahler manifolds and associated 2-forms both of which we dub \emph{separable}, since they are determined by single-variable functions in a preferred orthogonal coordinate system. We will show that separable K\"ahler metrics naturally unify the known concepts of product-toric, Calabi-toric and ortho-toric K\"ahler metrics, which are intimately related to the existence of  a Hamilton 2-form~\cite{Apostolov2001TheGO, Legendre} (we explore this connection in Section \ref{sec:ham-2-forms}). This section is written to be self-contained as it may be of interest more widely.

\subsection{Separable toric K\"ahler metrics}

Consider a toric K\"ahler manifold $B$ with K\"ahler metric $h$ and K\"ahler form $J$.  In symplectic coordinates $(x_i, \phi^i)$ this takes the form (\ref{h_toric}), \eqref{hessiang}, \eqref{KahlerForm-gen}, where the toric Killing fields are $m_i=\partial_{\phi^i}$ which we assume to have  $2\pi$-periodic orbits, and the associated moment maps are $x_i$.  In order to introduce the concept of separability it turns out to be more convenient to use an orthogonal coordinate system, as follows. 

Let $\xi, \eta$  be nonconstant functions that are invariant under the toric symmetry and orthogonal in the sense,
\begin{equation}\label{eq:orthogonality-gen}
\td\xi\cdot\td\eta=0\,,
\end{equation}
where $\cdot$ denotes the inner product defined by $h$.
It follows that we can use $(\xi, \eta, \phi^i)$ as local coordinates on $B$, so the  Jacobian of the transformation $(\xi, \eta)\mapsto (x_1(\xi, \eta), x_2(\xi, \eta))$, 
 \begin{equation}\label{warp}
\wrpp  :=\langle\partial_{\xi}x,\partial_{\eta}x\rangle\neq 0 \,,
\end{equation}
where we use the notation $\langle \alpha,\beta\rangle=\epsilon^{ij}\alpha_{i} \beta_{j}$ and the alternating symbol is such that  $\epsilon^{12}=1$.  In the $(\xi, \eta)$ coordinates \eqref{eq:orthogonality-gen} is equivalent to $h_{\xi \eta}=0$, that is, it defines an orthogonal coordinate system.  
It is useful to denote the other metric components by
\begin{equation}\label{eq:calFG-def}
\mathcal{F}(\xi,\eta):=h_{\xi\xi}^{-1}\,,\qquad\mathcal{\cg}(\xi,\eta):=h_{\eta\eta}^{-1}\,.
\end{equation}
 Changing coordinates $(\xi,\eta)\mapsto (x_1, x_2)$, the $x_i x_j$ components of the K\"ahler metric  \eqref{h_toric} give an expression for $G^{ij}$ and its inverse is
\begin{equation}\label{Gdd-general}
G_{ij}=\cf\partial_{\xi}x_{i}\partial_{\xi}x_{j}+\cg\partial_{\eta}x_{i}\partial_{\eta}x_{j} \; .
\end{equation}
Using this, it follows that the K\"ahler metric and K\"ahler form in such an orthogonal coordinate system take the simple form
\begin{align}\label{Kahler-orthogonal}
h &=\cf^{-1}\td\xi^{2}+\cg^{-1}\td\eta^{2}+\cf\sigma_{\xi}^{2}+\cg\sigma_{\eta}^{2}\,, \\
J &=   \td \xi \wedge \sigma_\xi+ \td\eta\wedge \sigma_\eta  \; , \label{Kahlerform_orthogonal}
\end{align}
where  we have defined the 1-forms 
\begin{equation}\label{sigma-1forms}
\sigma_{\xi} :=\partial_{\xi}x_{i}\td\phi^{i} 
\,,\qquad\sigma_{\eta} :=\partial_{\eta}x_{i}\td\phi^{i} \; .  
\end{equation}
Using the natural orthonormal frame for the metric, we can easily write down a basis of SD 2-forms, 
\begin{align}\label{SDtoric-2-forms}
I^{(1)}& =\td\xi\wedge\sigma_{\xi}-\td\eta\wedge\sigma_{\eta}\,,\nonumber \\
I^{(2)} &=(\cf\cg)^{1/2}\sigma_{\xi}\wedge\sigma_{\eta}+(\cf\cg)^{-1/2}\td\xi\wedge\td\eta\,, \nonumber \\
I^{(3)} &=(\cf/\cg)^{-1/2}\td\xi\wedge\sigma_{\eta}+(\cf/\cg)^{1/2}\td\eta\wedge\sigma_{\xi} \; ,
\end{align}
and ASD 2-forms (recall $J=J^{(1)}$), 
\begin{align}\label{ASDtoric-2-forms}
 J^{(1)} &=\td\xi\wedge\sigma_{\xi}+\td\eta\wedge\sigma_{\eta}\,,\nonumber \\
J^{(2)} &=(\cf\cg)^{1/2}\sigma_{\xi}\wedge\sigma_{\eta}-(\cf\cg)^{-1/2}\td\xi\wedge\td\eta\,,\nonumber \\
J^{(3)} &=(\cf/\cg)^{-1/2}\td\xi\wedge\sigma_{\eta}-(\cf/\cg)^{1/2}\td\eta\wedge\sigma_{\xi},
\end{align}
where the orientation is $\td \xi \wedge \td \eta\wedge \sigma_\xi\wedge \sigma_\eta$, both of which satisfy the quaternion algebra \eqref{quaternion}.

Now recall that $G^{ij}$ is the Hessian of the symplectic potential \eqref{hessiang} and writing this in orthogonal coordinates we find that \eqref{eq:orthogonality-gen} is equivalent  to a PDE for the symplectic potential, namely, 
\begin{equation}\label{symplectic-PDE1}
\partial_{\xi}\partial_{\eta}g=  (\partial_{\xi}\partial_{\eta}x_{i})  \partial^{i}g\, , 
\end{equation}
and the other components give
\begin{equation}\label{symplectic-calFG1}
\cf^{-1}=\partial_{\xi}^{2}g-(\partial_{\xi}^{2}x_{i})\partial^{i}g 
\,,\qquad\cg^{-1}=\partial_{\eta}^{2}g-(\partial_{\eta}^{2}x_{i} )\partial^{i}g  \; .
\end{equation}
All we have done so far is rewritten a general toric K\"ahler metric in orthogonal coordinates on the 2d orbit space, which is always possible. We are now ready to introduce the concept of separability.

\begin{definition}\label{separable-def}
A toric K\"ahler manifold $(B, h, J)$ is separable if there exists an orthogonal coordinate system $(\xi, \eta)$, as in \eqref{eq:orthogonality-gen}, such that the moment maps $x_i$ of the toric Killing fields satisfy,
\begin{equation}\label{eq:separable-def}
\partial_{\xi}^{2}x_{i}=0\,,\qquad\partial_{\eta}^{2}x_{i}=0\,, 
\end{equation}
that is,  the moment maps are linear in each of $\xi, \eta$.
\end{definition}

Integrating the above we can write 
\begin{equation}\label{moments-separable}
x_{i}=c_{i}+a_{i}\xi+\tilde{a}_{i}(\xi+\eta)+b_{i}\xi\eta\,,
\end{equation}
for some constants $c_{i}, a_{i}, \tilde{a}_{i}, b_{i}$. This definition reduces the freedom in the choice of orthogonal coordinates $\xi$ and $\eta$ to just affine transformations
\begin{equation}\label{separable-transformations}
\xi\to K_{\xi}\xi+C_{\xi}\,,\qquad\eta\to K_{\eta}\eta+C_{\eta}\,,
\end{equation}
where $K_{\xi}\neq 0, K_{\eta}\neq 0$ and $C_\xi, C_\eta$ are constants, 
as well as exchanging their roles
\begin{equation}\label{eq:exchange-xieta}
\xi\to \eta\,,\qquad\eta\to \xi\,.
\end{equation}
The constants in \eqref{moments-separable} also transform under \eqref {separable-transformations}  as,
\begin{align}\label{constants-transformation}
b_{i} & \to K_{\xi}^{-1}K_{\eta}^{-1}b_{i}\,,\nonumber \\
\tilde{a}_{i} & \to K_{\eta}^{-1}(\tilde{a}_{i}-K_{\xi}^{-1}C_{\xi}b_{i})\,,\nonumber \\
a_{i} & \to K_{\xi}^{-1}a_{i}+(K_{\xi}^{-1}-K_{\eta}^{-1})\tilde{a}_{i}+K_{\xi}^{-1}K_{\eta}^{-1}(C_{\xi}-C_{\eta})b_{i}\,,\nonumber \\
c_{i} & \to c_{i}-K_{\xi}^{-1}C_{\xi}a_{i}-(K_{\xi}^{-1}C_{\xi}+K_{\eta}^{-1}C_{\eta})\tilde{a}_{i}+K_{\xi}^{-1}K_{\eta}^{-1}C_{\xi}C_{\eta}b_{i}\, ,
\end{align}
while under \eqref{eq:exchange-xieta} as,
\begin{equation}\label{exchange-constants}
b_{i}\to b_{i}\,,\qquad\tilde{a}_{i}\to\tilde{a}_{i}+a_{i}\,,\qquad a_{i}\to-a_{i}\,,\qquad c_{i}\to c_{i}\,.
\end{equation}
These transformations will allow us classify separable metrics into three types.

To this end, consider the  Jacobian \eqref{warp} for a separable K\"ahler metric. The moment maps $x_i$ are given by \eqref{moments-separable} so we find
\be
\Omega =  \langle a, \tilde a \rangle+\langle a, b \rangle \xi+ \langle \tilde{a}, b \rangle (\xi-\eta)  \;,    \label{Om_separable}
\ee
in particular, note that $\partial_\xi \Omega$ and $\partial_\eta\Omega$ are both constant.  There are three qualitatively different cases to consider depending on whether both, one or none of the constants  $\partial_\xi \Omega$ and $\partial_\eta\Omega$ vanish. For each case, there is a canonical choice such that exactly one of the vectors  $ a_{i}, \tilde{a}_{i}, b_{i}$ vanishes identically.  First, if both constants vanish $\partial_{\xi}\wrpp=\partial_{\eta}\wrpp=0$ then $\langle a, b \rangle=\langle \tilde a, b \rangle=0$ and $\langle a, \tilde a \rangle \neq 0$, which implies $b_{i}=0$ (since $b_i$ can't be parallel to both $a_i$ and $\tilde a_i$). Secondly, if one of the constants vanish, without loss of generality we may assume $\partial_\eta\Omega=0$  due to \eqref{eq:exchange-xieta},  so $\partial_{\eta}\wrpp=\langle \tilde{a}, b \rangle = 0$ and $\partial_\xi \Omega=\langle a, b \rangle \neq 0$, which means that  $b \neq 0$ and $\tilde{a}$ is parallel to $b$, so from the second line of \eqref{constants-transformation} we can always fix $\tilde{a}_{i}=0$. Thirdly, if both $\partial_{\xi} \wrpp = \langle a, b \rangle +\langle \tilde a, b \rangle \neq 0$ and $\partial_{\eta}\wrpp =- \langle \tilde a, b \rangle \neq 0$, then $\tilde a, b$ are linearly independent, and hence from the third line of \eqref{constants-transformation} we can always fix  $a_{i}=0$. 

These three cases are summarised in Table \ref{table:sep-def}, where the Jacobian  is written as
\begin{equation}\label{eq:wrpp}
\wrpp=\nrm\sgm\,,
\end{equation}
for  a  nonzero constant $\nrm$ and function $\sgm(\xi, \eta)$ which are given for each case in Table \ref{table:sep-def}.  We now show that these cases correspond to \emph{product-toric} (PT), \emph{Calabi-toric} (CT) and \emph{ortho-toric}  (OT), respectively, therefore justifying the names in Table \ref{table:sep-def}~\cite{Legendre}. In particular, we show that in orthogonal coordinates $\xi, \eta$, each case is completely characterised in terms of two functions of a single variable $F(\xi), G(\eta)$ (this is a generalisation of the Calabi-toric case in Proposition 1 in \cite{Lucietti:2022fqj}).  

\begin{table}[h!]
\centering
\begin{tabular}{|c|c|c|c|c|}
\hline 
Class & Definition & \multicolumn{3}{c|}{Canonical choice}\tabularnewline
\hline 
\hline 
Product-toric (PT) & $\partial_{\xi}\wrpp = \partial_{\eta}\wrpp =0$ & $b_{i}=0$ & $N=\langle a,\tilde{a}\rangle \neq 0$ & $\Sigma= 1$\tabularnewline
\hline 
Calabi-toric (CT) & $\partial_{\xi}\wrpp \neq 0$, $\partial_{\eta}\wrpp=0 $  & $\tilde{a}_{i}=0$ & $N=\langle a,b\rangle \neq 0$ & $\Sigma= \xi$\tabularnewline
\hline 
Ortho-toric (OT) & $\partial_{\xi}\wrpp \neq 0$, $\partial_{\eta}\wrpp\neq 0$  & $a_{i}=0$ & $N=\langle\tilde{a},b\rangle \neq 0$ & $\Sigma=\xi-\eta$ \tabularnewline
\hline 
\end{tabular}
\caption{The three classes of separable toric K\"ahler structures.}
\label{table:sep-def}
\end{table}

\begin{prop}
\label{prop_separable}
Any separable toric K\"ahler metric can be written in the form
\begin{equation}\label{eq:metric-separable}
h=\sgm\Big(\frac{\td\xi^{2}}{F(\xi)}+\frac{\td\eta^{2}}{G(\eta)}\Big)+\frac{1}{\sgm}\Big(F(\xi)\sigma_{\xi}^{2}+G(\eta)\sigma_{\eta}^{2}\Big)\,,
\end{equation}
where
\begin{align}
  \sigma_\xi= \begin{cases}  \td \psi \\ \td \psi+ \eta \td \varphi \\  \td \psi+ \eta \td \varphi \end{cases} , \qquad   \sigma_\eta= \begin{cases}  \td \varphi  & \qquad \text{PT}\\  \xi  \td \varphi &  \qquad \text{CT} \\  \td \psi+ \xi \td \varphi & \qquad \text{OT}  \end{cases},  \label{eq:sigmaSigma_defs}
\end{align} 
 the K\"ahler form is
\begin{align}
J= \begin{cases} \td\big(\xi\td\psi+\eta\td\varphi\big) & \qquad PT  \\ \td\big(\xi\td\psi+\xi\eta\td\varphi\big) &  \qquad CT \\ \td\big((\xi+\eta)\td\psi+\xi\eta\td\varphi \big)&  \qquad OT\end{cases} ,
\end{align}
with the cases PT, CT, OT and the corresponding function $\Sigma$ defined in Table  \ref{table:sep-def}, and $\partial_\psi, \partial_\phi$ are a basis for the toric Killing fields.\footnote{Note that the Killing fields $\partial_\psi$ and $\partial_\varphi$ do not necessarily have closed orbits.}
\end{prop}

\begin{proof}By definition, for a separable toric K\"ahler structure the moment maps are given by \eqref{moments-separable}.  
The orthogonality condition \eqref{symplectic-PDE1} now yields a PDE for the symplectic potential which for the canonical cases listed in Table   \ref{table:sep-def} takes the form
\begin{equation}\label{symplectic-PDE2}
\sgm\partial_{\xi}\partial_{\eta}g-(\partial_{\eta}\sgm)\partial_{\xi}g-(\partial_{\xi}\sgm)\partial_{\eta}g=0\,,
\end{equation}
where $\Sigma$ is defined as in Table \ref{table:sep-def} and we have used that $b_{i}\partial^{i}=\sgm^{-1}(\partial_{\eta}\sgm\partial_{\xi}+\partial_{\xi}\sgm\partial_{\eta})$. The general solution to \eqref{symplectic-PDE2} in each case can be written as
\begin{equation}\label{eq:sympl-sep}
g=\sgm^{3}\Big[\partial_{\xi}\Big(\frac{\aff(\xi)}{\sgm^{2}}\Big)+\partial_{\eta}\Big(\frac{\bff(\eta)}{\sgm^{2}}\Big)\Big]\,,
\end{equation}
where $\aff(\xi)$ and $\bff(\eta)$ are arbitrary functions.\footnote{To prove this for the OT case it helps to first differentiate \eqref{symplectic-PDE2} with respect to both $\xi$ and $\eta$.} We can evaluate the functions $\cf$ and $\cg$ appearing in  \eqref{Kahler-orthogonal} using \eqref{symplectic-calFG1} to obtain
\begin{equation}\label{cfcg-FG}
\cf=\sgm^{-1}F(\xi)\,,\qquad\cg=\sgm^{-1}G(\eta)\,,
\end{equation}
where we have defined
\begin{equation}
F(\xi):=\frac{1}{\aff'''(\xi)}\,,\qquad G(\eta):=\frac{1}{\bff'''(\eta)}\,.
\end{equation}
Inserting \eqref{cfcg-FG} into  \eqref{Kahler-orthogonal} we obtain \eqref{eq:metric-separable} as required.  Finally, defining 
\begin{align}
\label{angletransf}
\psi= \begin{cases}   (a_{i}+\tilde{a}_{i})\phi^{i} \\  a_{i}\phi^{i} \\ \tilde{a_{i}}\phi^{i} \end{cases}, \qquad \varphi= \begin{cases} \tilde{a}_{i}\phi^{i} \\ b_{i}\phi^{i} \\ b_{i}\phi^{i} \end{cases}, \qquad \begin{array}{c} PT \\ CT \\ OT \end{array}  \; , 
\end{align}
we deduce the claimed form  for the 1-forms  \eqref{sigma-1forms} and the K\"ahler form \eqref{Kahlerform_orthogonal}.
\end{proof}

It is useful to note that for separable metrics the Gram matrix of Killing fields \eqref{Gdd-general} simplifies to,
\be
\label{Gij-sep}
G_{ij}= \frac{F(\xi)}{\Sigma} (a+\tilde a + b \eta)_i (a+\tilde a + b \eta)_j +  \frac{G(\eta)}{\Sigma} (\tilde a + b \xi)_i (\tilde a + b \xi)_j   \; ,
\ee
where we have used \eqref{moments-separable} and \eqref{cfcg-FG}.  Thus, its determinant takes the simple form
\begin{align}
\det G_{ij} = N F(\xi) G(\eta)  \; ,   \label{detG}
\end{align}
where $N$ is the constant in each of the three cases given in Table \ref{table:sep-def}.
Furthermore, the functions $F(\xi)$ and $G(\eta)$ can be expressed in terms of invariants on the K\"ahler base via the following projections:
\begin{align}
\label{proj-PT}
F(\xi)  &= \frac{G_{ij} \tilde a_k \tilde a_l \epsilon^{ik}\epsilon^{jl} }{\langle \tilde a, a \rangle^2} , \qquad G(\eta)=\frac{G_{ij}(a_k +\tilde a_k)(a_l +\tilde a_l) \epsilon^{ik}\epsilon^{jl} }{\langle \tilde a, a \rangle^2} , \qquad \text{if PT}  \\
\label{proj-CT}
\frac{F(\xi)}{\xi}  &= \frac{G_{ij} b_k b_l \epsilon^{ik}\epsilon^{jl} }{\langle a, b \rangle^2} , \quad - \frac{\eta F(\xi)}{\xi} = \frac{G_{ij} a_k b_l \epsilon^{ik}\epsilon^{jl} }{\langle a, b \rangle^2} \; , \quad  \xi G(\eta) +\frac{\eta^2 F(\xi)}{\xi} =  \frac{G_{ij} a_k a_l \epsilon^{ik}\epsilon^{jl} }{\langle a, b \rangle^2}
, \quad \text{if CT} \\
\label{proj-OT}
F(\xi)  &= \frac{G_{ij}( \tilde a_k +b_k \xi) (\tilde a_l+b_l\xi) \epsilon^{ik}\epsilon^{jl} }{(\xi - \eta) \langle \tilde a, b \rangle^2} , \qquad G(\eta)=\frac{G_{ij}(\tilde a_k+b_k\eta)(\tilde a_l+b_l\eta) \epsilon^{ik}\epsilon^{jl} }{(\xi -\eta) \langle \tilde a, b \rangle^2} , \qquad \text{if OT}
\end{align}
where $\epsilon^{ij}$ is the alternating symbol with $\epsilon^{12}=1$ and recall the three cases are defined in Table \ref{table:sep-def}.  

\subsection{Separable 2-forms}

We  now introduce a class of separable 2-forms on toric K\"ahler manifolds.  Note that we will  define this independently to the notion of separable K\"ahler metrics introduced in the previous section, that is, we do not assume Definition \ref{separable-def}.    

We start with a toric closed 2-form as in \eqref {mu-i-introduction} which, in orthogonal coordinates \eqref{Kahler-orthogonal}, reads
\begin{equation}\label{eq:F-expr}
F=\td\big(\wrpp^{-1}(\mu_{\xi}\sigma_{\xi}+\mu_{\eta}\sigma_{\eta})\big)\,,
\end{equation}
where we have defined
\begin{equation}\label{eq:muxieta-def}
\mu_{\xi}:=\langle\mu,\partial_{\eta}x\rangle 
\,,\qquad\mu_{\eta}:=-\langle\mu,\partial_{\xi}x\rangle 
\,.
\end{equation}
These transform under the gauge transformations \eqref{U1-gauge} as,
\begin{equation}\label{eq:gauge-mu-xieta}
\mu_{\xi}\to\mu_{\xi}+\partial_{\eta}\langle\alpha,x\rangle\,,\qquad\mu_{\eta}\to\mu_{\eta}-\partial_{\xi}\langle\alpha,x\rangle\,.
\end{equation}
Now, by  Lemma \ref{lem_JinvF_toric} it follows that closed, $J$-invariant, 2-forms that are invariant under toric symmetry, are a subclass of toric closed 2-forms as introduced in Definition \ref{def-toric-2forms}. The condition   \eqref{eq_JinvF} required for $J$-invariance becomes
 \begin{equation}\label{J-invariance-orthogonal}
\cf\langle\partial_{\xi}x,\partial_{\xi}\mu\rangle+\cg\langle\partial_{\eta}x,\partial_{\eta}\mu\rangle=0\, ,
\end{equation}
which  is also equivalent to the (local) existence of a  potential $\Lambda$ defined by \eqref{eq_Lambda} in terms of which 
\begin{equation}\label{mu-xieta-Lambda}
\mu_{\xi}=\wrpp\cf\partial_{\xi}\Lambda\,,\qquad\mu_{\eta}=\wrpp\cg\partial_{\eta}\Lambda\, .
\end{equation}
It is useful to note that \eqref{J-invariance-orthogonal} can be written as
\begin{equation}\label{eq_J-inv-alt}
\cf\partial_{\xi}\mu_{\eta}-\cg\partial_{\eta}\mu_{\xi}=\langle\cf\partial_{\xi}^{2}x-\cg\partial_{\eta}^{2}x,\mu\rangle\,.
\end{equation}
Therefore, on a K\"ahler surface that is separable with respect to orthogonal coordinates $(\xi, \eta)$, the r.h.s. of \eqref{eq_J-inv-alt} vanishes due to \eqref{eq:separable-def}  and hence using \eqref{cfcg-FG} the $J$-invariance condition reduces to 
\begin{equation}\label{sep-on-sep}
F(\xi)\partial_{\xi}\mu_{\eta} - G(\eta)\partial_{\eta}\mu_{\xi}=0\, ,
\end{equation}
while further using \eqref{eq:wrpp}, equations \eqref{mu-xieta-Lambda} reduce to
\begin{equation}\label{mu-xieta-Lambda-2}
\mu_{\xi}=\nrm F(\xi)\partial_{\xi}\Lambda\,,\qquad\mu_{\eta}=\nrm G(\eta)\partial_{\eta}\Lambda\,.
\end{equation}
We are now ready to define a class of separable 2-forms.

\begin{definition}\label{separable1form-def}
A  toric closed 2-form $F$, on a toric K\"ahler surface, is separable if there exist orthogonal coordinates as in \eqref{eq:orthogonality-gen} such that
\begin{equation}\label{eq:separable1form-def}
\langle\partial_{\xi}x,\partial_{\xi}\mu\rangle=0\,,\qquad\langle\partial_{\eta}x,\partial_{\eta}\mu\rangle=0\, ,
\end{equation}
where $\mu_i$ are the moment maps associated to $F$, that is, $\iota_{m_i} F= -\td \mu_i$.
\end{definition}

The motivation for this definition will become apparent shortly\footnote{It is worth noting that  separability of 2-forms also admits a definition in terms of holomorphic coordinates $y^i+i
\phi^i$, that is,  \eqref{eq:separable1form-def} can be written as
$\partial_{\xi}y^{i}\partial_{\eta}\mu_{i}=0$ and $\partial_{\eta}y^{i}\partial_{\xi}\mu_{i}=0$.}. First, observe that according to this definition, a separable toric closed 2-form is necessarily $J$-invariant, in particular, both terms in the $J$-invariance condition \eqref{J-invariance-orthogonal} vanish separately.  Therefore, separable 2-forms are a special subclass of $J$-invariant 2-forms, as illustrated  in Figure \ref{fig:2-forms}.  The archetypal separable 2-form on a generic toric K\"ahler surface is the K\"ahler form itself, since in this case $\mu_{i}=x_{i}$ which trivially solves \eqref{eq:separable1form-def}. Observe that the definition of separable 2-forms \eqref{eq:separable1form-def} is invariant under the gauge transformations \eqref{U1-gauge}.

\begin{figure}[h!]
\centering
   \includegraphics[width=0.50\textwidth]{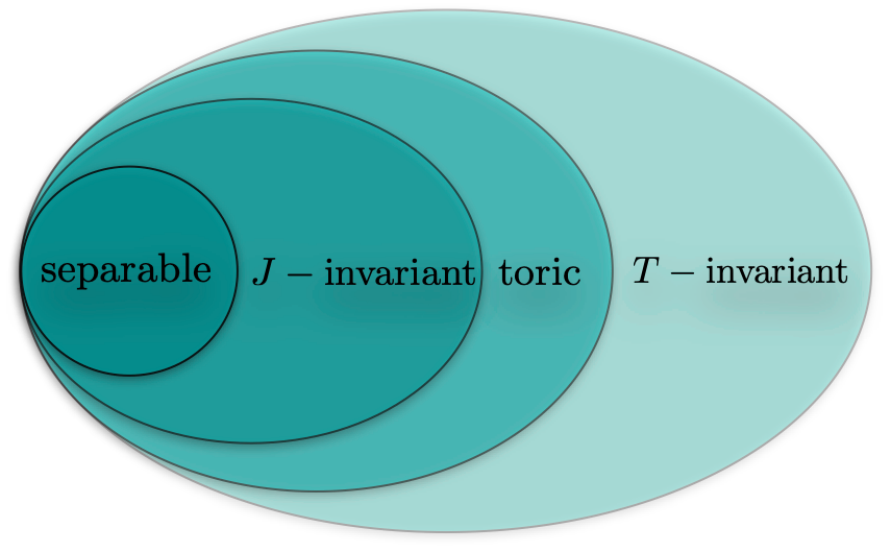}
\caption{Various classes of closed 2-forms we use in this paper. Notice that separable 2-forms are a subclass of the $J$-invariant ones.}
\label{fig:2-forms}
\end{figure}

 The following result  shows that if the concepts of a separable 2-form and metric are compatible, the 2-form can also be described in terms of functions of a single variable, thus justifying the use of the term ``separable''. 
 
 \begin{prop}\label{sep-Lambda}
Let $(B, h, J)$ be a separable toric K\"ahler surface as in Proposition \ref{prop_separable}.  Then, a toric closed 2-form $F$, that is separable with respect to the same orthogonal coordinates $(\xi, \eta)$,  takes the form
\be
F= \td \left( \frac{1}{N\Sigma} (\mu_{\xi}\sigma_{\xi}+\mu_{\eta}\sigma_{\eta})\right)\,,
\ee
where
\begin{equation}\label{eq:sep-ansatz}
\mu_{\xi}=\mu_{\xi}(\xi)\,,\qquad \mu_{\eta}=\mu_{\eta}(\eta)\,.
\end{equation}
Furthermore, the potential for $F$ defined by \eqref{eq_Lambda} is additively separable,
\begin{equation}\label{eq:Lambda-sep}
\Lambda(\xi,\eta)=\Lambda_{\xi}(\xi)+\Lambda_{\eta}(\eta)\,,
\end{equation}
for functions $\Lambda_{\xi}(\xi)$ and $\Lambda_{\eta}(\eta)$. Conversely, on a separable toric K\"ahler surface \eqref{eq:sep-ansatz} or \eqref{eq:Lambda-sep} imply that $F$ is separable.
\end{prop}
\begin{proof} Recall that the functions $\mu_\xi, \mu_\eta$ are defined by \eqref{eq:muxieta-def}. Differentiating these by $\eta$ and $\xi$ respectively, using  metric separability \eqref{eq:separable-def} and 2-form separability \eqref{eq:separable1form-def}, we deduce
\begin{equation}\label{eq:sep-ansatz2}
\partial_{\eta}\mu_{\xi}=0 \,, \qquad \partial_{\xi}\mu_{\eta}=0\,,
\end{equation}
which establishes \eqref{eq:sep-ansatz}.\footnote{This statement is invariant under the gauge transformations \eqref{eq:gauge-mu-xieta} (recall $\partial_{\xi}^{2}x_{i}=\partial_{\eta}^{2}x_{i}=0$). } The form for $F$ then follows from \eqref{eq:F-expr} and \eqref{eq:wrpp}. Next, using \eqref{mu-xieta-Lambda-2},  we see that both equations in \eqref{eq:sep-ansatz2} reduce to 
\begin{equation}
\partial_{\xi}\partial_{\eta}\Lambda=0\,,
\end{equation}
thus proving \eqref{eq:Lambda-sep}.

Conversely, if $\partial_\eta \mu_\xi=0$ and $\partial_\xi \mu_\eta=0$, then \eqref{eq:muxieta-def} implies that $\langle\partial_{\xi}x,\partial_{\xi}\mu\rangle=0$ and $\langle\partial_{\eta}x,\partial_{\eta}\mu\rangle=0$ and hence $F$ is separable.
\end{proof}

Observe that this lemma shows that a  2-form that is separable, with respect to the same orthogonal coordinates for which a K\"ahler metric is separable, corresponds to one where both terms on the l.h.s. of \eqref{sep-on-sep} vanish separately.

We have already noted below Lemma \ref{lem_JinvF_toric} that the K\"ahler potential is the $\Lambda$-potential for the K\"ahler form. We therefore deduce the following corollary.
\begin{cor}\label{sep-Kahler-pot}
For a toric K\"ahler surface that is separable with respect to orthogonal coordinates $(\xi, \eta)$, the K\"ahler potential is additively separable, 
\begin{equation}
\ck(\xi,\eta)=\ck_{\xi}(\xi)+\ck_{\eta}(\eta)\,.
\end{equation}
\end{cor}

We can verify this statement by directly evaluating the Legendre transform \eqref{Kahler-Legendre} of the symplectic potential for separable metrics given in \eqref{eq:sympl-sep} and we find 
\begin{equation}
\ck=\begin{cases}
\Big(\xi^{3}\big(\xi^{-2}\aff(\xi)\big)'\Big)'+\Big(\eta^{3}\big(\eta^{-2}\bff(\eta)\big)'\Big)' & \text{for PT}\,,\\
\partial_{\xi}\sgm\,\xi^{2}\Big(\xi^{2}\big(\xi^{-2}\aff(\xi)\big)'\Big)'+\partial_{\eta}\sgm\,\eta^{2}\Big(\eta^{2}\big(\eta^{-2}\bff(\eta)\big)'\Big)' & \text{for CT and OT}\, , 
\end{cases}
\end{equation}
where without loss of generality we set $c_{i}=0$ in \eqref{moments-separable}.\footnote{Constant shifts of $x_{i}$ correspond to K\"ahler transformations.} 
It is interesting to note that the K\"ahler potential for CT type depends only on $\xi$.

The next result gives another generic example of a separable 2-form.

 \begin{lemma}\label{sep-Ricci}
Consider a toric K\"ahler surface separable with respect to orthogonal coordinates $(\xi, \eta)$ as in Proposition \ref{prop_separable}. The Ricci 2-form $\mathcal{R}=\td P$ is also separable with respect to $(\xi, \eta)$ and given by 
\begin{equation}
P=-\frac{1}{2\sgm}\big(F'(\xi)\sigma_{\xi}+G'(\eta)\sigma_{\eta}\big)\, . \label{eq_P_sep}
\end{equation}
Furthermore, the $\Lambda$-potential for $\mathcal{R}$ is given by
\begin{equation}
\Lambda_{\text{Ricci-form}} =-\frac{1}{2}\log F(\xi)G(\eta) \,.   \label{eq_Lambda_Ricci}
\end{equation}
\end{lemma}
\begin{proof}
The Ricci form potential for a toric K\"ahler metric in symplectic coordinates is given by \eqref{eq_ricci_toric}, so in particular its $\Lambda$-potential is given by $\Lambda= -\frac{1}{2} \log \det G_{ij}$. The matrix $G_{ij}$ can be computed from \eqref{eq:metric-separable} and \eqref{angletransf}, which in all cases gives \eqref{eq_Lambda_Ricci} (up to an irrelevant additive constant). Then, using \eqref{mu-xieta-Lambda}, \eqref{cfcg-FG}, and \eqref{eq:F-expr} gives \eqref{eq_P_sep} as claimed.  Therefore, by Proposition \ref{sep-Lambda} the Ricci form is separable.
\end{proof}

\subsection{Hamiltonian 2-forms}\label{sec:ham-2-forms}

In the preceding two subsections, we introduced the notion of separable metrics and 2-forms on toric K\"ahler surfaces. In this subsection, we will establish a connection between metric separability and the theory of Hamiltonian 2-forms~\cite{Apostolov2001TheGO}.  In fact, in~\cite{Apostolov2001TheGO, Legendre} it has been shown that 
any toric K\"ahler metric that admits a Hamiltonian 2-form must be precisely one of PT, CT or OT, and therefore separable according to our definition. We  deduce the following  theorem which gives a geometrical characterisation of our notion of separability.

 \begin{theorem}\label{hamiltonian-2-forms-theorem}
A toric K\"ahler metric is separable if and only if it admits a Hamiltonian 2-form. 
\end{theorem}

We will provide a self-contained proof of one direction of this theorem, namely, that any separable toric K\"ahler surface admits a Hamiltonian 2-form. We will show this by an explicit calculation and in fact determine all possible Hamiltonian 2-forms on such geometries.

A Hamiltonian 2-form on a K\"ahler surface $(B, h, J)$ may be defined as a closed $J$-invariant 2-form $\Psi$ that satisfies~\cite{Apostolov2001TheGO},
\begin{equation}\label{def-ham-2form}
\nabla_{a}\Psi_{bc}=\frac{2}{3}\big((\partial_{a}\sigma)J_{bc}-J_{a[b}\partial_{c]}\sigma-h_{a[b}J_{c]}^{\,\,d}\partial_{d}\sigma\big)\,,
\end{equation}
where\footnote{This should not be confused with other quantities denoted by $\sigma$ in this paper.} 
\begin{equation}
\sigma :=\frac{1}{4}\Psi_{ab}J^{ab}\, .
\end{equation}
Notice that the K\"ahler form $J$ always satisfies \eqref{def-ham-2form} and hence is trivially a  Hamiltonian 2-form. However, in certain cases \eqref{def-ham-2form} may admit more interesting solutions which we refer to as non-trivial Hamiltonian 2-forms.

In the case of a toric K\"ahler surface, if we assume $\Psi$ also has toric symmetry, then by Lemma \ref{lem_JinvF_toric} we can write 
\begin{equation}
\Psi=\td(\mu_{i}\td\phi^{i})\,,
\end{equation}
for  some moment maps $\mu_i$. Since by definition a Hamiltonian 2-form is $J$-invariant, $\mu_{i}$ are required to satisfy \eqref{eq_JinvF}. A computation then reveals that  in symplectic coordinates equation \eqref{def-ham-2form} is equivalent to
\begin{equation}\label{h2form-toric}
E_{\,\,\,\,k}^{ij}=0\,,
\end{equation}
where we have defined
\begin{align}
E_{\,\,\,\,k}^{ij} :&=\partial^{i}\partial^{j}\mu_{k}-\frac{1}{2}G_{\ell p}\partial^{p}G^{ij}\partial^{\ell}\mu_{k}+\frac{1}{2}G_{kp}\partial^{p}G^{i\ell}\partial^{j}\mu_{\ell} \nonumber \\ 
&-\frac{1}{3}(\delta_{k}^{j}\partial^{i}\partial^{\ell}\mu_{\ell}+\frac{1}{2}\delta_{k}^{i}\partial^{j}\partial^{\ell}\mu_{\ell}+\frac{1}{2}G^{ij}G_{kp}\partial^{p}\partial^{\ell}\mu_{\ell})\, .
\end{align}
In order to find non-trivial Hamiltonian 2-forms, we need to solve \eqref{h2form-toric} together with \eqref{eq_JinvF}. We stress that this system does not admit non-trivial solutions for general toric K\"ahler metrics, indeed, by Theorem \ref{hamiltonian-2-forms-theorem} it admits non-trivial solutions precisely for separable K\"ahler metrics. This follows from the integrability properties of \eqref{h2form-toric} which have been studied in~\cite{Apostolov2001TheGO,Legendre}. 

Our goal here is to assume separability of the K\"ahler metric and explicitly solve \eqref{h2form-toric} together with \eqref{eq_JinvF}. With this assumption, we find the following components of \eqref{h2form-toric},
\begin{align}\label{h2form-toric-comb}
0=(2\partial_{\eta}x_{i}\partial_{\xi}x_{p}-\partial_{\xi}x_{i}\partial_{\eta}x_{p})\partial_{\xi}x_{j}\epsilon^{pk}E_{\,\,\,\,k}^{ij} & =\nrm G(\eta)\partial_{\xi}\partial_{\eta}^{2}\Lambda\,,\nonumber \\
0=-(2\partial_{\xi}x_{i}\partial_{\eta}x_{p}-\partial_{\eta}x_{i}\partial_{\xi}x_{p})\partial_{\eta}x_{j}\epsilon^{pk}E_{\,\,\,\,k}^{ij} & =\nrm F(\xi)\partial_{\xi}^{2}\partial_{\eta}\Lambda\,.
\end{align}
In order to arrive at the r.h.s., we have used the fact that a separable metric can be PT, CT or OT and written the result in a unified way. Moreover, we have traded $\mu_{i}$ for $\mu_{\xi}$ and $\mu_{\eta}$ through \eqref{eq:muxieta-def} and expressed the latter in terms of the potential $\Lambda$ as in \eqref{mu-xieta-Lambda-2} exploiting the $J$-invariance of the Hamiltonian 2-forms. The general solution to \eqref{h2form-toric-comb} is 
\begin{equation}
\Lambda=p\xi\eta+\Lambda_{\xi}(\xi)+\Lambda_{\eta}(\eta)\,,
\end{equation}
for functions $\Lambda_{\xi}(\xi), \Lambda_{\eta}(\eta)$ and a constant $p$. Note that for $p=0$ the above expression reduces to \eqref{eq:Lambda-sep}. We can then use again \eqref{mu-xieta-Lambda-2} to find
\begin{equation}\label{eq:mu-aux}
\mu_{\xi}=\nrm pF(\xi)\eta+\hmu_{\xi}(\xi)\,,\qquad \mu_{\eta}= \nrm pG(\eta)\xi+\hmu_{\eta}(\eta)\,,
\end{equation}
where $\hmu_{\xi}(\xi):=\nrm F(\xi)\Lambda_{\xi}'(\xi)$ and $\hmu_{\eta}(\eta):=\nrm G(\eta)\Lambda_{\eta}'(\eta)$. Therefore, by Proposition \ref{sep-Lambda}, we deduce that the corresponding Hamiltonian 2-form is itself separable if and only if $p=0$.\footnote{In this subsection, when we say that the Hamiltonian 2-form is separable, we always mean separability with respect to the orthogonal coordinates for which  the metric is separable.}   It remains to solve the rest of the equations in \eqref{h2form-toric} for $\hmu_{\xi}(\xi)$ and $\hmu_{\eta}(\eta)$. The independent components are
\begin{align}\label{h2form-toric-ind}
0 & =E_{1}:=6\frac{F(\xi)}{G(\eta)}\partial_{\xi}x_{i}\partial_{\xi}x_{j}\partial_{\xi}x_{p}\epsilon^{pk}E_{\,\,\,\,k}^{ij}\,,\nonumber \\
0 & =E_{2}:=3\partial_{\xi}x_{i}\partial_{\xi}x_{j}\partial_{\eta}x_{p}\epsilon^{pk}E_{\,\,\,\,k}^{ij}\,.
\end{align}
It turns out that for a generic separable toric K\"ahler metric we have $p=0$ and hence the Hamiltonian 2-form is also separable. However, for certain specific choices of separable metrics (i.e. specific functions $F(\xi), G(\eta)$), which we dub \emph{exceptional}, there exist Hamiltonian 2-forms with $p\neq 0$, that is, they also admit non-separable Hamiltonian 2-forms. We present these exceptional cases in the Appendix  \ref{sec:separability-appendix}.

We now consider separable solutions to \eqref{h2form-toric}, i.e. we set $p=0$, which as we will see does not impose any restrictions on the functions $F(\xi)$ and $G(\eta)$. Our results are summarised by the following proposition which gives an explicit proof of one direction in Theorem \ref{hamiltonian-2-forms-theorem}.  Observe that this shows that the space of non-trivial Hamiltonian 2-forms on generic separable K\"ahler surfaces is 1-dimensional.
 
 \begin{prop}\label{sep-ham}
The most general Hamiltonian 2-form on a separable toric K\"ahler surface, that is separable with respect to the same orthogonal coordinates $(\xi, \eta)$, is given by
\begin{equation}
\Psi=\gamma\,J+\delta\,\Psi_{\ast}\,,
\end{equation}
where $\gamma$ and $\delta$ are constants, $J$ is the K\"ahler form (as in Proposition \ref{prop_separable}) and $\Psi_{\ast}$ is a non-trivial Hamiltonian 2-form given in Table \ref{table:sep-ham}.
 \begin{table}[h!]
\centering
\begin{tabular}{|c|c|}
\hline 
Class & Separable Hamiltonian 2-form\tabularnewline
\hline 
\hline 
PT & $\Psi_{\ast}=\td(\xi\td\psi-\eta\td\varphi)$\tabularnewline
\hline 
CT & $\Psi_{\ast}=\td\big(\xi^{2}(\td\psi+\eta\td\varphi)\big)$\tabularnewline
\hline 
OT & $\Psi_{\ast}=\td\Big(\frac{\xi^{3}}{\xi-\eta}(\td\psi+\eta\td\varphi\big)-\frac{\eta^{3}}{\xi-\eta}(\td\psi+\xi\td\varphi\big)\Big)$\tabularnewline
\hline 
\end{tabular}
\caption{Non-trivial Hamiltonian 2-forms for separable toric K\"ahler metrics.}
\label{table:sep-ham}
\end{table}
\end{prop}
\begin{proof} 
As we have already mentioned $J$ is always a Hamiltonian 2-form so we will focus on $\Psi_{\ast}$. We will look for solutions to \eqref{h2form-toric-ind} by examining each of the cases PT, CT and OT separately. We will also utilise the gauge transformations \eqref{eq:gauge-mu-xieta}. Further notice that since $p=0$ in \eqref{eq:mu-aux}, we have  ${\mu}_{\xi}=\mu_{\xi}(\xi)$ and ${\mu}_{\eta}=\mu_{\eta}(\eta)$.

For PT geometries we have 
\begin{equation}
0=E_{1}=-\mu_{\eta}''(\eta)\,,\qquad0=E_{2}=-\mu_{\xi}''(\xi)\,,
\end{equation}
with solution $\mu_{\xi}(\xi)=\gamma_{0}+\gamma_{1}\xi$ and $\mu_{\eta}(\eta)=\delta_{0}+\delta_{1}\eta$, where $\gamma_{0,1}, \delta_{0,1}$ are constants. The constant terms $\gamma_0, \delta_0$ can be fixed to zero using  \eqref{eq:gauge-mu-xieta}.

For CT geometries we have 
\begin{equation}
0=\partial_{\xi}^{2}(\xi^{2}E_{1})=-2\mu_{\eta}''(\eta)\,,\qquad0=\partial_{\eta}^{3}(\xi^{2}E_{2})=-2\xi\mu_{\xi}''''(\xi)\,,
\end{equation}
so $\mu_{\xi}(\xi)$ is a cubic polynomial and $\mu_{\eta}(\eta)$
is a linear one. Then $0=E_{1}=E_{2}$ further constrains the coefficients
of these polynomials such that $\mu_{\xi}(\xi)=\gamma_{1}\xi+\gamma_{2}\xi^{2}+\gamma_{3}\xi^{3}$
and $\mu_{\eta}(\eta)=\delta_{0}-\gamma_{1}\eta$, where $\gamma_i, \delta_0$ are constants. Using the gauge transformations \eqref{eq:gauge-mu-xieta}
we can fix  $\mu_{\xi}(\xi)=\gamma_{2}\xi^{2}+\gamma_{3}\xi^{3}$ and
$\mu_{\eta}(\eta)=0$. 

For OT geometries we have 
\begin{equation}
0=\partial_{\xi}^{3}\big((\xi-\eta)^{2}E_{1}\big)=2\mu_{\xi}''''(\xi)\,,\qquad0=\partial_{\eta}^{3}\big((\xi-\eta)^{2}E_{2}\big)=-2(\xi-\eta)\mu_{\eta}''''(\eta)\,,
\end{equation}
so both $\mu_{\xi}(\xi)$ and $\mu_{\eta}(\eta)$ are cubic polynomials.
Then $0=E_{1}=E_{2}$ further implies $\mu_{\xi}(\xi)=\gamma_{0}+\gamma_{1}\xi+\gamma_{2}\xi^{2}+\gamma_{3}\xi^{3}$
and $\mu_{\eta}(\eta)=-\gamma_{0}-\gamma_{1}\eta-\gamma_{2}\eta^{2}-\gamma_{3}\eta^{3}$, for constants $\gamma_i$.
Using the gauge transformations   \eqref{eq:gauge-mu-xieta} we can fix $\mu_{\xi}(\xi)=\gamma_{2}\xi^{2}+\gamma_{3}\xi^{3}$
and $\mu_{\eta}(\eta)=-\gamma_{2}\eta^{2}-\gamma_{3}\eta^{3}$ .

The resulting Hamiltonian 2-forms follow from Proposition \ref{sep-Lambda} and are given in Table \ref{table:sep-ham}.
\end{proof}

\section{Separable supersymmetric solutions}
\label{sec:separablesusy}

In this section we will introduce the concept of separable supersymmetric solutions to five-dimensional gauged supergravity.

\begin{definition}
\label{def:sep-sol}
A timelike supersymmetric toric solution to five-dimensional minimal gauged supergravity, possibly coupled to $n-1$ vector multiplets, is separable (or PT, CT, OT) if:
\begin{itemize}
\item  The toric K\"ahler base $(B, h, J)$ is separable (PT, CT, OT) with respect to orthogonal coordinates $(\xi, \eta)$ (Definition \ref{separable-def}).
\item  The magnetic fields $F^I$  are also separable (PT, CT, OT) with respect to the orthogonal coordinates $(\xi, \eta)$  (see Definition \ref{separable1form-def}). In minimal supergravity the magnetic field is completely determined by the Ricci form and so by Lemma \ref{sep-Ricci}  this condition is automatically satisfied.
\end{itemize}
\end{definition}

We will first investigate supersymmetric solutions that are timelike and separable outside a regular horizon with compact locally spherical cross-sections. We will find that the only allowed type of separable toric K\"ahler metric compatible with such a horizon is Calabi-toric.  Then, we will perform a detailed analysis of Calabi-toric supersymmetric solutions and prove uniqueness of the known black hole within this class.

 \subsection{Near-horizon analysis}\label{sec:no-gos}
We now  examine the constraints imposed by the  existence of a smooth horizon on timelike supersymmetric solutions with separable toric-K\"ahler base metrics.  A key point in our analysis is that the $\hat{\eta}$-dependence of the moment maps  to leading order in GNC takes the form
\begin{equation}\label{linearform}
x_i = \frac{(\text{linear in }\hat{\eta})}{\hat{H}(\hat{\eta})^{1/3}}\;  \lambda + O(\lambda^2)\,,
\end{equation}
as can be seen from \eqref{xi-NH-expl}, and the Gram matrix of Killing fields takes the form
\be
G_{ij} =  \frac{ (\text{quadratic or cubic in } \hat{\eta})}{\hat{H}(\hat{\eta})^{1/3}}\;  \lambda + O(\lambda^2)\,,  \label{Gij-eta-dep}
\ee
which follows from \eqref{Gdd-NH}.

 We will examine the three classes of separable metric in turn.
 \begin{lemma}\label{product-no-go}
Consider a  supersymmetric toric solution to STU supergravity that is timelike outside a compact horizon with (locally) $S^3$ cross-sections. Then it cannot have a PT K\"ahler base.
\end{lemma}
\begin{proof}
Recall that this case corresponds to $b=0$ (see Table \ref{table:sep-def}). Then inverting \eqref{moments-separable} gives
\begin{equation}\label{xieta-prod}
\xi=\frac{\langle x-c,\tilde{a}\rangle}{\langle a,\tilde{a}\rangle}\,,\qquad\eta=-\frac{\langle x-c,a+\tilde{a}\rangle}{\langle a,\tilde{a}\rangle}\,.
\end{equation}
Therefore  \eqref{linearform} implies
\begin{equation}\label{standard-expansion}
\xi=\xi_{0}+\frac{\xi_{1}(\hat{\eta})}{\hat{H}(\hat{\eta})^{1/3}}\lambda+O(\lambda^{2})\,,\qquad\eta=\eta_{0}+\frac{\eta_{1}(\hat{\eta})}{\hat{H}(\hat{\eta})^{1/3}}\lambda+O(\lambda^{2})\,,
\end{equation}
where $\xi_{0}=\langle \tilde{a},c\rangle/\langle a,\tilde{a}\rangle$ , $\eta_{0}=-\langle a+\tilde{a},c\rangle/\langle a,\tilde{a}\rangle$ are constants and $\xi_{1}(\hat{\eta})$, $\eta_{1}(\hat{\eta})$ are linear functions of $\hat\eta$. 
Next, from \eqref{proj-PT} and the fact that \eqref{Gdd-NH} implies $G_{ij}=O(\lambda)$, we learn that $F(\xi)=O(\lambda)$ and $G(\eta)=O(\lambda)$. Therefore, combining with \eqref{standard-expansion} we deduce $F(\xi_{0})=G(\eta_{0})=0$. Expanding \eqref{Gij-sep} to linear order in $\lambda$ we find,
\begin{equation}
G_{ij}=\frac{\lambda}{\hat{H}(\hat{\eta})^{1/3}}\Big(F'(\xi_{0})(a_{i}+\tilde a_i)(a_{j}+\tilde a_j) \xi_{1}(\hat{\eta})
+G'(\eta_{0})\tilde a_{i} \tilde a_{j}\eta_{1}(\hat{\eta})\Big)+O(\lambda^{2})\,.
\end{equation}
The factor in the brackets is a linear function of $\hat\eta$ which contradicts the explicit form of $G_{ij}$ given by \eqref{Gij-eta-dep}.
\end{proof}
 
 \begin{lemma}\label{orthotoric-no-go}
 Consider a  supersymmetric toric solution to STU supergravity that is timelike outside a compact horizon with (locally) $S^3$ cross-sections. Then it cannot have an orthotoric K\"ahler base.
\end{lemma}
\begin{proof}
Recall we can set $a_{i}=0$ in \eqref{moments-separable} (see Table \ref{table:sep-def}). We find that $\xi$ and $\eta$ are given by the solutions of the quadratic equation
\begin{equation}
\langle\tilde{a},b\rangle\chi^{2}+\langle b,x-c\rangle\chi+\langle\tilde{a},x-c\rangle=0\,.
\end{equation}
Thus, without loss of generality we can write,
\begin{equation}\label{xi-eta-solution}
\xi=\chi_{+}\,,\qquad\eta=\chi_{-}\,,\qquad\text{with}\qquad\chi_{\pm}=\frac{-\langle b,x-c\rangle\pm\sqrt{\langle b,x-c\rangle^{2}-4\langle\tilde{a},b\rangle\langle\tilde{a},x-c\rangle}}{2\langle\tilde{a},b\rangle}\,.
\end{equation}
From \eqref{xi-eta-solution} it is clear that the horizon is mapped to a point $(\xi_{0},\eta_{0})=(\chi_{+}^{0},\chi_{-}^{0})$ with 
\begin{equation}
\chi_{\pm}^{0} :=\frac{\langle b,c\rangle\pm\sqrt{\langle b,c\rangle^{2}+4\langle\tilde{a},b\rangle\langle\tilde{a},c\rangle}}{2\langle\tilde{a},b\rangle}\,, 
\end{equation}
The analysis splits into two cases.

Let us first consider the case where the discriminant is nonvanishing, $\langle b,c\rangle^{2}+4\langle\tilde{a},b\rangle\langle\tilde{a},c\rangle\protect\neq0$, or equivalently $\xi_{0}\neq\eta_{0}$. In this case the expansions of $\xi$ and $\eta$ in $\lambda$ are as in \eqref{standard-expansion} where again $\xi_{1}(\hat{\eta})$, $\eta_{1}(\hat{\eta})$ are again linear due to \eqref{linearform}. Since $\xi_{0}\neq\eta_{0}$, from \eqref{proj-OT} and \eqref{Gij-eta-dep} we infer $F(\xi)=O(\lambda)$ and $G(\eta)=O(\lambda)$, so 
\begin{equation}
F(\xi_{0})=G(\eta_{0})=0\,.
\end{equation}
We then find that \eqref{Gij-sep} implies
\begin{equation}
G_{ij}=\frac{\lambda}{\hat H(\hat\eta)^{1/3}}\Big(\frac{F'(\xi_{0})}{\xi_{0}-\eta_{0}}(\tilde{a}_{i}+b_{i}\eta_{0})(\tilde{a}_{j}+b_{j}\eta_{0})\xi_{1}(\hat{\eta})+\frac{G'(\eta_{0})}{\xi_{0}-\eta_{0}}(\tilde{a}_{i}+b_{i}\xi_{0})(\tilde{a}_{j}+b_{j}\xi_{0})\eta_{1}(\hat{\eta})\Big)+O(\lambda^{2})\,,
\end{equation}
which is incompatible with \eqref{Gij-eta-dep} since the factor in the brackets is linear in $\hat\eta$.

We now turn to the case with vanishing discriminant, $\langle b,c\rangle^{2}+4\langle\tilde{a},b\rangle\langle\tilde{a},c\rangle\protect=0$, or equivalently $\xi_{0}=\eta_{0}$.  
The near-horizon expansions of \eqref{xi-eta-solution} that follow from \eqref{linearform}  are now of the form 
\begin{equation}\label{xieta-ortho-special}
\xi=\xi_{0}+\xi_{1/2}(\hat{\eta})\sqrt{\lambda}+\xi_{1}(\hat{\eta})\lambda+O(\lambda^{3/2})\,,\qquad\eta=\eta_{0}+\eta_{1/2}(\hat{\eta})\sqrt{\lambda}+\eta_{1}(\hat{\eta})\lambda+O(\lambda^{3/2})\, , 
\end{equation}
for some functions $\xi_{1/2}(\hat\eta)$ etc., whose explicit form we will not need.  Next, it  is useful to note that for an OT metric \eqref{Gij-sep} implies that
\begin{equation}\label{orthotoric-proj}
\frac{G_{ij}\epsilon^{ik}\epsilon^{jl}b_{k}b_{l}}{\langle\tilde{a},b\rangle^{2}}=\frac{F(\xi)+G(\eta)}{\xi-\eta}\,.  
\end{equation}
Therefore, the near-horizon behaviour  \eqref{Gij-eta-dep} implies that both $F(\xi)/(\xi-\eta)= O(\lambda)$ and   $G(\eta)/(\xi-\eta)= O(\lambda)$ (since both of these are non-negative functions).  Furthermore, \eqref{xieta-ortho-special} implies $\xi-\eta= O(\sqrt{\lambda})$ so we deduce that in fact $F(\xi)=O(\lambda^{3/2})$ and $G(\eta)=O(\lambda^{3/2})$.  In turn, using \eqref{detG} this implies $\det G_{ij}= O(\lambda^3)$ which contradicts the form \eqref{Gij-eta-dep} since the quadratic/cubic prefactor never vanishes identically.
\end{proof}

We pause to emphasise that both Lemma \ref{product-no-go} and Lemma \ref{orthotoric-no-go} also apply to minimal gauged supergravity since this is a consistent truncation of  the STU theory. In fact, due to the near-horizon uniqueness theorem in the minimal theory~\cite{Kunduri:2006uh}, both of these lemmas hold under the weaker hypothesis that the cross-sections are compact (since they have to be locally $S^3$ in this theory).  Furthermore, in our previous paper, we showed that in the minimal theory for any solution of this type with a Calabi-toric base, the solution must be locally isometric to the CCLP black hole, see~\cite[Theorem 1]{Lucietti:2022fqj}. We have therefore now established Theorem \ref{minimal-theorem}.   

It remains to study the near-horizon form of such solutions with a Calabi-toric base in the STU theory. From Proposition \ref{prop_separable}, we can write any Calabi-toric surface as 
\begin{align}
\label{CT}
 h &= \rho \left( \frac{\td\rho^2}{F(\rho)} + \frac{\td \eta^2}{G(\eta)} \right) + \frac{F(\rho)}{\rho} (\td \psi+\eta \td \varphi)^2+ \rho G(\eta) \td \varphi^2  \\
 J&= \td (\rho  (\td \psi+\eta \td \varphi))  \; ,
\end{align}
where in order to be consistent with the minimal theory we have set $\xi=\rho$.  The canonical choice $\tilde{a}_{i}=0$ breaks the shift-freedom of $\xi$ in \eqref{separable-transformations} and the residual transformations with $C_{\xi}=0$ act in the coordinates as
\begin{equation}\label{Calabi-symmetries}
\rho\to K_{\rho}\rho\,,\qquad\psi\to K_{\rho}^{-1}(\psi-C_{\eta}K_{\eta}^{-1}\varphi)\,,\qquad\eta\to K_{\eta}\eta+C_{\eta}\,,\qquad\varphi\to K_{\rho}^{-1}K_{\eta}^{-1}\varphi\,,
\end{equation}
on the functions $F(\rho), G(\eta)$ as
\begin{equation}
F(\rho)\to K_{\rho}^{3}F(\rho)\,,\qquad G(\eta)\to K_{\rho}K_{\eta}^{2}G(\eta)\,. 
\end{equation}
and on the constants $a_{i}, b_{i}$ as
 \begin{equation}\label{ab-transform}
a_{i}\to K_{\rho}^{-1}(a_{i}-C_{\eta}K_{\eta}^{-1}b_{i})\,,\qquad b_{i}\to K_{\rho}^{-1}K_{\eta}^{-1}b_{i}\,.
\end{equation}
These freedoms in the choice of Calabi-coordinates will be useful in what follows. 

\begin{lemma}\label{Calabi-horizon}
 Consider a  supersymmetric toric solution to STU supergravity that is timelike outside a compact horizon with (locally) $S^3$ cross-sections.  If the K\"ahler base is Calabi-toric, then near the horizon,  Calabi coordinates $(\rho, \eta)$ are related to GNC $(\lambda, \hat\eta)$  by,
\begin{equation}\label{rhoeta-NH}
\rho=\frac{\ell\kappa}{4\hat{H}(\hat{\eta})^{1/3}}\lambda+O(\lambda^{2})\,,\qquad\eta=\hat{\eta}+O(\lambda)\,,
\end{equation}
where $H(\hat\eta)$ is given by \eqref{eq:H-function}, so in particular the horizon is at $\rho=0$. Furthermore, we can always choose Calabi coordinates such that,
\begin{equation}\label{eq:FG-Aux}
F(\rho)=\rho^{2}+O(\rho^3)\,,\qquad G(\eta)=(1-\eta^{2})\Delta_{1}(\eta)\, ,
\end{equation}
where the function $\Delta_1(\eta)$ is given by \eqref{Delta-functions}.
\end{lemma}

\begin{proof}
This proceeds in an identical fashion to the analogous lemma in the minimal theory~\cite{Lucietti:2022fqj}. For completeness we repeat it here.  Recall that for a CT base we may always set  $\tilde{a}_{i}=0$ in which case $\langle a, b \rangle \neq 0$, see Table \ref{table:sep-def}. Hence inverting \eqref{moments-separable}  we obtain
\begin{equation}\label{Calabi-coords}
\rho=\frac{\langle x-c,b\rangle}{\langle a,b\rangle}\,,\qquad\eta=-\frac{\langle x-c,a\rangle}{\langle x-c,b\rangle}\, .
\end{equation}
Therefore, the near-horizon expansion \eqref{Gij-eta-dep} together with  \eqref{proj-CT} imply that  as $\lambda\to 0$,
\begin{equation}\label{orders1}
\eta=O(1)\,,\qquad\frac{F(\rho)}{\rho}=O(\lambda)\,,\qquad\rho G(\eta)=O(\lambda)\,.
\end{equation}
These relations imply that $c_{i}=0$. 

To see this, suppose that $c_{i}$ does not vanish identically.  If $\langle b, c \rangle=0$ then $c_i$ is a nonzero multiple of $b_i$ and since $\langle a, b \rangle \neq 0$ it follows that $\langle a, c \rangle \neq 0$; then \eqref{Calabi-coords} and \eqref{linearform} imply that $\eta$ is singular at the horizon contradicting \eqref{orders1}.  We deduce that $\langle b, c \rangle \neq 0$.  Then, expanding \eqref{Calabi-coords} using \eqref{xi-NH-expl} we find,
\begin{align}\label{rhoeta-fake}
\rho & =\frac{\langle b,c\rangle}{\langle a,b\rangle}+\frac{1}{\langle a,b\rangle}\Big(\frac{1-\hat{\eta}}{\ca^{2}}b_{2}-\frac{1+\hat{\eta}}{\cb^{2}}b_{1}\Big)\frac{\ell\kappa\lambda}{4\hat{H}(\hat{\eta})^{1/3}}+O(\lambda^{2})\,,\nonumber \\
\eta & =-\frac{\langle a,c\rangle}{\langle b,c\rangle}+\frac{\langle a,b\rangle}{\langle b,c\rangle^{2}}\Big(\frac{1-\hat{\eta}}{\ca^{2}}c_{2}-\frac{1+\hat{\eta}}{\cb^{2}}c_{1}\Big)\frac{\ell\kappa\lambda}{4\hat{H}(\hat{\eta})^{1/3}}+O(\lambda^{2})\, ,
\end{align}
where we have used that $a_i-b_i\langle a, c \rangle/\langle b, c \rangle=- c_i \langle a, b \rangle /\langle b, c \rangle$, 
showing that in Calabi-coodinates  the horizon maps to a point $(\rho,\eta)=(\rho_{0}, \eta_{0})$ where $\rho_0\neq 0, \eta_0\neq 0$. From \eqref{orders1} it then follows that we have $F(\rho_{0})=G(\eta_{0})=0$ and expanding \eqref{Gij-sep} to linear order in $\lambda$ we find 
\begin{align}
G_{ij}  &=\Bigg[\frac{G'(\eta_{0})b_{i}b_{j}}{\langle b,c\rangle}\Big(\frac{1-\hat{\eta}}{\ca^{2}}c_{2}-\frac{1+\hat{\eta}}{\cb^{2}}c_{1}\Big) \Bigg.  \nonumber \\  &\qquad \Bigg.  +\frac{\langle a,b\rangle^{2}F'(\rho_{0})c_{i}c_{j}}{\langle b,c\rangle^{3}}\Big(\frac{1-\hat{\eta}}{\ca^{2}}b_{2}-\frac{1+\hat{\eta}}{\cb^{2}}b_{1}\Big)\Bigg]\frac{\ell\kappa\lambda}{4\hat{H}(\hat{\eta})^{1/3}}+O(\lambda^{2})\, .
\end{align} 
The term in square brackets has linear $\hat{\eta}$-dependence which contradicts  \eqref{Gij-eta-dep}. We  deduce that our assumption that $c_i \neq 0$ must be false and hence
\begin{equation}\label{eq:c-zero}
c_{i}=0 , 
\end{equation}
as claimed. 

Now, the relation between the Calabi coordinates $(\rho,\eta)$ and the GNC $(\lambda,\hat\eta)$ near the horizon that follows from \eqref{Calabi-coords} and \eqref{xi-NH-expl} becomes, 
\begin{align}\label{eq:rhoeta-GNC}
\rho & =\frac{1}{\langle a,b\rangle}\Big(\frac{1-\hat{\eta}}{\ca^{2}}b_{2}-\frac{1+\hat{\eta}}{\cb^{2}}b_{1}\Big)\frac{\ell\kappa}{4\hat{H}(\hat{\eta})^{1/3}}\lambda+O(\lambda^{2})\,,\nonumber \\
\eta & =-\frac{\ca^{2}a_{1}(1+\hat{\eta})-\cb^{2}a_{2}(1-\hat{\eta})}{\ca^{2}b_{1}(1+\hat{\eta})-\cb^{2}b_{2}(1-\hat{\eta})}+O(\lambda)\,,
\end{align}
which in particular implies that the horizon in Calabi-coordinates is given by $\rho=0$.  Inverting we also deduce that $\lambda$ is a smooth function of $\rho$ at the horizon.

To complete the proof of our lemma, we need to show that there exist functions $F(\rho)$
and $G(\eta)$ such that \eqref{Gij-sep} reproduces \eqref{Gdd-NH} at $O(\lambda)$. From \eqref{orders1} and \eqref{eq:rhoeta-GNC} we see that  $F(\rho) = O(\lambda^2)$ and $G(\eta)=O(1)$ are smooth functions of $\lambda$ at the horizon. Therefore, we can write, 
\begin{equation}
F(\rho)=F_{2}\rho^{2}+O(\lambda^{3})\,,\qquad G(\eta)=G_{0}(\hat{\eta})+O(\lambda)\,,
\end{equation}
where
\begin{equation}\label{eq:F2-G0-Aux}
F_{2}=\frac{1}{2}F''(0)\,,\qquad 
G_{0}(\hat{\eta})=G\Big(-\frac{\ca^{2}a_{1}(1+\hat{\eta})-\cb^{2}a_{2}(1-\hat{\eta})}{\ca^{2}b_{1}(1+\hat{\eta})-\cb^{2}b_{2}(1-\hat{\eta})}\Big)\,.
\end{equation}
One can now check that \eqref{Gdd-NH} and  \eqref{Gij-sep} match at $O(\lambda)$ if and only if 
\begin{equation}\label{matching-conditions}
\frac{b_{2}}{b_{1}}=-\frac{\ca^{2}}{\cb^{2}}\,,\qquad F_{2}=\frac{2}{a_{1}\ca^{2}+a_{2}\cb^{2}}\,,\qquad G_{0}(\hat{\eta})=-\frac{(1-\hat{\eta}^{2})\Delta_{1}(\hat{\eta})}{F_{2}\ca^{2}\cb^{2}b_{1}b_{2}}\,.
\end{equation}
We can now exploit the freedom in the choice of Calabi type coordinates \eqref{ab-transform} to fix 
\begin{equation}\label{ab-fix}
a_{1}=-b_{1}=\ca^{-2}\,,\qquad a_{2}=b_{2}=\cb^{-2}\,,
\end{equation}
which thus fixes
\begin{equation}\label{F2Gcub}
F_{2}=1\,,\qquad G_{0}(\hat{\eta})=(1-\hat{\eta}^{2})\Delta_{1}(\hat{\eta})\,,
\end{equation}
and  \eqref{eq:rhoeta-GNC} simplifies to \eqref{rhoeta-NH}. The second equation in \eqref{eq:F2-G0-Aux} now reduces to $G_{0}(\hat{\eta})=G(\hat\eta)$ which therefore determines the function $G(\eta)$ as claimed. 
\end{proof}

It is worth  noting that in the Calabi-coordinates of Lemma \ref{Calabi-horizon} the constant $\nrm =\langle a, b \rangle$ is fixed by the near-horizon geometry parameters \eqref{ab-fix} to be simply
\begin{equation}
\nrm=2\ca^{-2}\cb^{-2}\,.   \label{eq:N}
\end{equation}

\subsection{Black hole uniqueness theorem in STU supergravity}\label{sec:STU-uniqueness}

We now turn to the proof of our main result, Theorem \ref{stu-theorem}, which is the first black hole uniqueness theorem in STU gauged supergravity.  This is a generalisation of the corresponding result in  minimal gauged supergravity, namely Theorem \ref{minimal-theorem}, which itself is a generalisation of the main theorem in~\cite{Lucietti:2022fqj}. The main assumption of Theorem \ref{minimal-theorem} is the separability of the K\"ahler metric on the base space.  In particular, in the minimal theory we did not need to make any assumptions on the Maxwell field, since this is completely determined by the K\"ahler  data. This is not surprising since both the metric and the Maxwell field are part of the gravitational multiplet, the only multiplet present in the minimal theory. However, in the STU model, the gravitational multiplet is supplemented with two vector multiplets, so naturally we find that an analogous uniqueness theorem requires making further assumptions on the Maxwell fields. In this section, we will analyse the supersymmetry constraints  for the class of separable toric solutions (see Definition \ref{def:sep-sol}) and use the near-horizon analysis in Section \ref{sec:no-gos} to prove the uniqueness theorem. In fact, the near-horizon analysis, which is summarised by Lemmas \ref{product-no-go}, \ref{orthotoric-no-go}, \ref{Calabi-horizon}, reveals that the only type of separable solution compatible with a horizon is Calabi-toric, so our analysis will restrict to this class.

\subsubsection{Calabi-toric supersymmetry solutions with horizons}\label{sec:susy-separable}

Consider a timelike supersymmetric toric solution that is Calabi-toric as in Definition \ref{def:sep-sol}. Thus both the K\"ahler metric $h$ and magnetic fields $F^I$ are separable (CT) and we write the metric in the form \eqref{CT}.  Recall that  by Lemma \ref{J-inv-Maxwell} the magnetic fields are closed, $J$-invariant, 2-forms invariant under the toric symmetry, therefore, by Proposition \ref{sep-Lambda}, if $F^I$ are also separable (CT), their gauge fields can be written as
\begin{equation}
A^{I}=\frac{1}{N}\Big(\frac{\mu_{\rho}^{I}(\rho)}{\rho}(\td\psi+\eta\td\varphi)+\mu_{\eta}^{I}(\eta)\td\varphi\Big)\, ,  \label{eq:CT-gaugefield}
\end{equation}
for functions $\mu_\rho^{I}(\rho), \mu_\eta^{I}(\eta)$ where $N=\langle a, b \rangle$ is a constant (see Table \ref{table:sep-def}). Note that the gauge transformations \eqref{eq:gauge-mu-xieta} for the CT case act as
\be
\label{eq:gauge-mu-xieta-CT}
\mu^I_\rho(\rho) \to \mu^I_\rho(\rho) + \langle \alpha^I, b \rangle \rho , \qquad \mu_\eta^I(\eta)\to \mu^I_\eta(\eta)- \langle \alpha, a \rangle- \langle \alpha, b \rangle \eta  \; .
\ee
We will show that 
separability imposes strong restrictions on supersymmetric solutions.

First, it is useful to note that for a Calabi-toric metric \eqref{CT} the basis of SD and ASD 2-forms \eqref{SDtoric-2-forms} and \eqref{ASDtoric-2-forms} become~\footnote{In order to avoid cluttered expressions, we will occasionally omit the argument of a function when there is no risk of confusion.}
\begin{align}
I^{(1)} &=\td\rho\wedge(\td\psi+\eta\td\varphi)-\rho\td\eta\wedge\td\varphi\,, \\
I^{(2)}&=\sqrt{FG}\td\psi\wedge\td\varphi+\frac{\rho}{\sqrt{FG}}\td\rho\wedge\td\eta\,,   \nonumber \\
I^{(3)}&=\rho\sqrt{\frac{G}{F}}\td\rho\wedge\td\varphi+\sqrt{\frac{F}{G}}\td\eta\wedge(\td\psi+\eta\td\varphi) \; ,  \nonumber
\end{align}
and
\begin{align}
\label{ASD-CT}
J^{(1)} &=\td\rho\wedge(\td\psi+\eta\td\varphi)+\rho\td\eta\wedge\td\varphi\,, \\
J^{(2)} &=\sqrt{FG}\td\psi\wedge\td\varphi-\frac{\rho}{\sqrt{FG}}\td\rho\wedge\td\eta\,,\nonumber \\
J^{(3)}&=\rho\sqrt{\frac{G}{F}}\td\rho\wedge\td\varphi-\sqrt{\frac{F}{G}}\td\eta\wedge(\td\psi+\eta\td\varphi) . \nonumber 
\end{align}
Now, evaluating the field strengths $F^{I}=\td A^{I}$ and comparing with \eqref{FI} with $\Theta^{I}=\Theta^{I}_{i}I^{(i)}$, 
we find  
\begin{equation}\label{Theta-0}
\Theta_{2}^{I}=\Theta_{3}^{I}=0\, ,
\end{equation}
and 
\begin{align}\label{thetaI}
\theta^{I} &:= \Theta_1^I=\frac{1}{2\nrm}\Big[\rho\Big(\frac{\mu_{\rho}^{I}(\rho)}{\rho^{2}}\Big)'-\frac{1}{\rho}\mu_{\eta}^{I}\phantom{}'(\eta)\Big]\,, \\
\Phi^{I} &=-\frac{1}{2\nrm\rho}\Big[\mu_{\rho}^{I}\phantom{}'(\rho)+\mu_{\eta}^{I}\phantom{}'(\eta)\Big]\,.  \label{PhiI}
\end{align}
This shows that for any CT supersymmetric solution the SD 2-forms $\Theta^{I}$ and the scalars $\Phi^{I}$  are fully fixed in terms of gauge field data.  It is convenient to note that the time rescalings \eqref{time-resc} for CT solutions can be realised by
\begin{equation}\label{time-resc-CT}
\rho\to K^{-1}\rho\,,\qquad F(\rho)\to K^{-2}F(\rho)\,,\qquad\mu_{\rho}^{I}(\rho)\to K^{-1}\mu_{\rho}^{I}(\rho)\,,
\end{equation}
with $\eta,\psi,\varphi,G(\eta)$ and $\mu_{\eta}^{I}(\eta)$ unchanged.

For solutions with horizons more information can be extracted from the near-horizon geometry. Recall that in Lemma \ref{Calabi-horizon} we showed that the near-horizon geometry completely fixes the function $G(\eta)$ appearing in the Calabi metric \eqref{CT}.  We will now show that an analogous result holds for the Maxwell field, that is, the function $\mu^I_\eta(\eta)$ is also fixed by the near-horizon geometry.

\begin{lemma} \label{CT-NH-gauge}
Consider a supersymmetric solution with a horizon as in Lemma \ref{Calabi-horizon}, with Maxwell fields that are separable with respect to the same Calabi-coordinates $(\rho, \eta)$. Then there is a gauge where
\begin{equation}\label{muI-eta}
\mu_\rho^I(\rho)= O(\rho^2) , \qquad \mu_{\eta}^{I}\phantom{}'(\eta)=-C^{IJK}\fil_{J}\frac{\nrm\ell^{2}}{6}(\ell\Delta_{2}(\eta)\zeta_K+\ck_{K})\, ,
\end{equation}
where $\Delta_2(\eta)$ is given by \eqref{Delta-functions}. 
\end{lemma}

\begin{proof}
For the CT case the function $\mu^I_\rho$ defined by \eqref{eq:muxieta-def} reduces to $\mu^I_\rho(\rho) = \langle \mu^I, b \rangle \rho$. Since $\mu_i^{I}$ and $\rho$ are smooth at the horizon  by \eqref{mu-NH}, \eqref{rhoeta-NH}, we deduce that $\mu_\rho^{I}(\rho)$ are smooth and  $\mu_\rho^I(\rho)= O(\rho)$. 
Furthermore, the gauge transformations \eqref{eq:gauge-mu-xieta-CT}  shift each $\mu_{\rho}^{I}\phantom{}'(\rho)$ by a constant and therefore we can fix a gauge so
\begin{equation}\label{gauge-fix-CT}
\mu_{\rho}^{I}\phantom{}'(0)=0\,.
\end{equation}
This gives the first equation in \eqref{muI-eta}.
Now, using the near-horizon expansions \eqref{rhoeta-NH}  in \eqref{PhiI} gives
\begin{equation}
\Phi^{I}=-\frac{2\hat{H}(\hat{\eta})^{1/3}}{\nrm\ell\kappa\lambda}\mu_{\eta}^{I}\phantom{}'(\hat{\eta})+O(1)\, .
\end{equation}
and comparing with the near-horizon expansion \eqref{Phi-NH} we deduce \eqref{muI-eta} as required.
\end{proof}

It remains to solve the rest of the supersymmetry conditions for the functions $F(\rho)$ and $\mu_\rho^{I}(\rho)$, subject to the near-horizon boundary conditions given in Lemma \ref{Calabi-horizon} and \ref{CT-NH-gauge}.

We start with  \eqref{Ricciscalar-1} and \eqref{ricci}, which  using $G(\eta)$ in \eqref{eq:FG-Aux}, $\mu^I_\eta(\eta)$ in  \eqref{muI-eta}, \eqref{thetaI} and \eqref{PhiI}, yield
\begin{equation}
\zeta_I \mu_{\rho}^{I}\phantom{}'(\rho)=\nrm\Big(1-\frac{1}{2}F''(\rho)\Big)\, ,
\end{equation}
and 
\be
\left(\frac{\zeta_I \mu_\rho^I(\rho)}{\rho^2}\right)'= \frac{N}{2 \rho^3} \left( 2 F'(\rho)- \rho F''(\rho)- 2 \rho \right) \; ,
\ee
respectively, which are equivalent to the single constraint,
\be
\zeta_I \mu^I_\rho(\rho)= \nrm\Big(\rho-\frac{1}{2}F'(\rho) \Big) .  \label{eq:isotropic}
\ee
The Maxwell equations \eqref{eq:maxwell-2} reduces to 
\begin{equation}\label{mxwl-CT}
\mathcal{E}_{I}:=\nabla^{2}\Big(\frac{\ell^{2}}{2}\fil_{J}\Phi^{J}\fil_{I}-G_{IJ}\Phi^{J}\Big)+\frac{\ell^{2}}{12}\fil_{I}\lambda_{1}+\frac{1}{\ell}C_{IJK}(\Phi^{J}\Phi^{K}-\theta^{J}\theta^{K})=0 \,,
\end{equation}
and 
$\lambda_{1}$, given by \eqref{eq:lambda1}, becomes
\begin{equation}
 \lambda_1= \frac{1}{2}\nabla^2 R  + \frac{4}{\ell^2} C_{IJK} \bar{X}^I  (\theta^J \theta^K- \Phi^J \Phi^K ) 
\, ,
\end{equation}
where we have used \eqref{Theta-0} and the definition $\theta^I:= \Theta_1^I$.  

Next, we recall that for a toric solution we showed that one can write $\omega$ as \eqref{om-toric}, so in Calabi-coordinates we must have $\omega=\omega_\psi\td \psi+ \omega_{\varphi} \td \varphi$ for functions $\omega_\psi, \omega_{\varphi}$ that depend only on $(\rho, \eta)$. This implies that $\td \omega$ does not have any $\psi \varphi$ components and hence from \eqref{dom}, \eqref{GMinus} and \eqref{ASD-CT}, we deduce that the function
\be
\lambda_2=0 \; .
\ee
Finally, the integrability condition for the existence of $\omega$, which is equivalent to  \eqref{lambda23eq}, now reduces to
\begin{equation}\label{lambda-CT}
\partial_{\rho}(\sqrt{FG}\lambda_{3})=-G\partial_{\eta}(\lambda_{1}+\scmb)\,,\qquad\partial_{\eta}(\sqrt{FG}\lambda_{3})=\frac{F}{\rho}\Big(\rho\partial_{\rho}(\lambda_{1}-\scmb)-2\scmb\Big)\,,
\end{equation}
where we have defined the combination
\begin{equation}
\scmb :=12\fil_{J}\Phi^{J}\fil_{I}\theta^{I}-24\ell^{-2}\Phi_{I}\theta^{I}\,.
\end{equation}
From \eqref{lambda-CT} we immediately get that the integrability condition for the existence of $\lambda_3$, which is equivalent to \eqref{eq:susyPDE}, is
\begin{equation}\label{int-CT}
\mathcal{E}:=\partial_{\rho}\Big[\frac{F}{\rho}\Big(\rho\partial_{\rho}(\lambda_{1}-\scmb)-2\scmb\Big)\Big]+\partial_{\eta}\Big(G\partial_{\eta}(\lambda_{1}+\scmb)\Big)=0\,.
\end{equation}
With $\theta^{I}$ and $\Phi^{I}$ expressed as in \eqref{thetaI} and  \eqref{PhiI}, $G(\eta)$ and $\mu_{\eta}^{I}\phantom{}'(\eta)$ given by \eqref{eq:FG-Aux} and \eqref{muI-eta} respectively, equations \eqref{eq:isotropic}, \eqref{mxwl-CT} and \eqref{int-CT}  become a system of $\eta$-dependent coupled ODEs for the functions $F(\rho)$ and $\mu_\rho^{I}(\rho)$.

\subsubsection{Uniqueness theorem}

We now have all the necessary ingredients to complete the proof of Theorem \ref{stu-theorem}. We start with the uniqueness of the K\"ahler base. In contrast to minimal supergravity this is coupled to the scalars and Maxwell fields and hence we must solve for these simultaneously.

\begin{prop}
\label{prop:base}
 Consider a supersymmetric toric solution to STU supergravity that is timelike and separable outside a smooth (analytic if $\ca^2=\cb^2$) horizon with compact (locally) $S^3$ cross-sections.  Then, the K\"ahler base is Calabi-toric \eqref{CT} where $G(\eta)$ and $F(\rho)$ are given by \eqref{eq:FG-Aux} and
 \be
 F(\rho)= \rho^2+ s \rho^3, 
 \ee
where the constant $s=
4/\ell^{2}$ or $s=0$, so in particular we can write,
\begin{equation}\label{Cal-KLR}
h=\frac{\td\rho^{2}}{\rho+s\rho^{2}}+(\rho+s\rho^{2})\sigma^{2}+\frac{\rho}{(1-\eta^{2})\Delta_{1}(\eta)}(\td\eta^{2}+\tau^{2})\,,\qquad J=\td(\rho\sigma)\,,
\end{equation}
where we have defined
\begin{align}
\sigma & :=\td\psi+\eta\td\varphi=\frac{1-\eta}{\ca^{2}}\td\phi^{1}+\frac{1+\eta}{\cb^{2}}\td\phi^{2}\,,\nonumber \\
\tau & :=-(1-\eta^{2})\Delta_{1}(\eta)\td\varphi=(1-\eta^{2})\Delta_{1}(\eta)\Big(\frac{\td\phi^{1}}{\ca^{2}}-\frac{\td\phi^{2}}{\cb^{2}}\Big)\,,
\end{align} 
and $\Delta_1(\eta)$ is given by \eqref{Delta-functions} and $0< \ca^{2},\cb^{2}<1$ are constants that parameterise the near-horizon geometry and satisfy \eqref{kappa2}.  Furthermore, the scalar fields are given by \eqref{PhiI} where the functions $\mu_\eta^I(\eta)$ and $\mu_\rho^I(\rho)$ are given by \eqref{muI-eta} and
\be
\mu_{\rho}^{I}(\rho)=-\frac{\ell  N s \bar{X}^I }{2}  \rho^{2} \;  ,   \label{muIrho}
\ee
so in particular, 
\be
\Phi^I = \frac{\ell^2}{12 \rho} C^{IJK} \zeta_J \big( (3 \rho s +\Delta_2(\eta))\ell \zeta_K+ \ck_K \big)  \; .   \label{eq:PhiIsol}
\ee

\end{prop}
\begin{proof}
We have already shown in Lemmas \ref{product-no-go}, \ref{orthotoric-no-go} and \ref{Calabi-horizon} that the only separable K\"ahler base compatible with a smooth horizon is Calabi-toric, with the functions $G(\eta)$ and $\mu_{\eta}^{I}\phantom{}'(\eta)$ determined by the near-horizon geometry where the latter follows from Lemma \ref{CT-NH-gauge}. Therefore we need to solve \eqref{eq:isotropic}, \eqref{mxwl-CT} and \eqref{int-CT} for $F(\rho), \mu^I_\rho(\rho)$. We will show that the \emph{only} solution to this system compatible with Lemma \ref{Calabi-horizon} and \ref{CT-NH-gauge} is
\begin{equation}\label{sol-KLR}
F(\rho)=\rho^{2}+F_{3}\rho^{3}\,,\qquad\mu_{\rho}^{I}(\rho)=-\frac{\ell  N F_{3} \bar{X}^I}{2} \rho^{2}\,,
\end{equation}
where $F_{3}$ an integration constant and $N$ is given by \eqref{eq:N}. Under the scaling freedom \eqref{time-resc-CT}  $F_{3}\to KF_{3}$ where $K\neq 0$ is constant, so we can use this to set $F_3=s$ where $s=0$ or $s=4/\ell^2$, which gives the claimed solution.  The explicit form of the base then follows from  \eqref{angletransf} and \eqref{ab-fix} which give,
\be
\psi = \frac{1}{\ca^2}\phi^1+ \frac{1}{\cb^2} \phi^2, \qquad \varphi= - \frac{1}{\ca^2}\phi^1+\frac{1}{\cb^2} \phi^2   \; .   \label{eq:psiphitransf}
\ee
As in the proof of the corresponding theorem in minimal supergravity~\cite{Lucietti:2022fqj}, we need to distinguish between the cases $\ca^{2}\neq\cb^{2}$ and $\ca^{2}=\cb^{2}$.

\subsubsection*{i) Case $\ca^{2}\neq\cb^{2}$}
By examining the explicit $\eta$-dependence of $\mathcal{E}$ in \eqref{int-CT} and $\mathcal{E}_{I}$ in \eqref{mxwl-CT} we find that both are polynomials in $\eta$ of degree two and one respectively. Using \eqref{eq:isotropic} and \eqref{ck-constraint} we find
\begin{equation}\label{int-aux}
\partial_{\eta}^{2}\mathcal{E}=-\frac{6(\ca^{2}-\cb^{2})^{2}}{\rho^{4}}\Big(\rho^{2}F''(\rho)-4\rho F'(\rho)+6F(\rho)\Big)\,,
\end{equation}
and
\begin{equation}\label{mxwl-aux}
\partial_{\eta}\mathcal{E}^{I}=-\frac{\ell(\ca^{2}-\cb^{2})}{4\rho}\Bigg(\frac{\bar{X}^IF'(\rho)}{3\rho}+\frac{\ca^{2}\cb^{2}\mu_{\rho}^{I}(\rho)}{\ell\rho}\Bigg)'\,.
\end{equation}
Hence \eqref{int-CT} implies the vanishing of \eqref{int-aux}  which can be easily solved to give
\begin{equation}
F(\rho)=F_{2}\rho^{2}+F_{3}\rho^{3}\,,
\end{equation}
with $F_{2}$ and $F_{3}$ integration constants. From Lemma \ref{Calabi-horizon} we deduce $F_{2}=1$ and hence $F(\rho)$ is given by the first equation in \eqref{sol-KLR}. Inserting then back in \eqref{mxwl-aux} we can solve for $\mu_{\rho}^{I}(\rho)$, which after fixing the relevant integration by using \eqref{gauge-fix-CT}, is given by the second equation in \eqref{sol-KLR}. It can now be checked that  \eqref{int-CT} and \eqref{mxwl-CT} are satisfied identically.  This completes the proof of \eqref{sol-KLR} for the case $\ca^{2}\neq\cb^{2}$.

\subsubsection*{ii) Case $\ca^{2}=\cb^{2}$}
In this special case, \eqref{int-aux} and \eqref{mxwl-aux} are automatically satisfied and therefore cannot be used to solve for $F(\rho)$ and $\mu^{I}(\rho)$. In fact, $\mathcal{E}$ and $\mathcal{E}^{I}$ ``lose'' all their $\eta$-dependence, and \eqref{int-CT} and \eqref{mxwl-CT} become ODEs for $F(\rho), \mu_\rho^I(\rho)$, which need to be solved together with \eqref{eq:isotropic}.  In order to present these ODEs in a more convenient form, we rescale the functions $F(\rho)$ and $\mu^{I}_{\rho}(\rho)$ and define
\begin{equation}
\mathcal{F}(\rho) :=\ca^{-2}\rho^{-2}F(\rho)\,,\qquad\nu^{I}(\rho):=\ca^{2}\mu^{I}_{\rho}(\rho)\,.
\end{equation}
In terms of these \eqref{eq:isotropic} is written as~\footnote{For this case of the proof we do not use the summation convention for the indices $I=1,2,3$ and write out all sums explicitly. We furthermore slightly abuse the position of the same indices by writing $\nu_{I}=\nu^{I}$.}
\begin{equation}\label{eq:F1-SU2} 
\ell^{-1} \sum_{I}\nu^{I}+(\rho^{2}\mathcal{F})' -\frac{2\rho}{\ca^{2}} =0\,,
\end{equation}
while 
\eqref{mxwl-CT} and \eqref{int-CT} read respectively,
\begin{align}\label{mxwl-aux-SU2}
\mathcal{E}_{I} & =\frac{1}{\rho}\Big[-\frac{\ell(\ca^{2}-1)}{18}(\ca^{2}\rho\mathcal{F}'-2)+\frac{\ca^{2}}{36\rho}\sum_{J=1}^{3}\mathcal{K}_{J}\nu^{J}+\frac{\ca^{4}\nu^{I}}{12\rho^{2}}(\rho^{2}\mathcal{F})' \nonumber \\
 & +\frac{\ca^{2}}{12\rho}\Big(\frac{\ca^{2}\nu_{I}^{2}}{\ell\rho}+\sum_{J,K=1}^{3}C^{IJK}\nu_{J}\big(\mathcal{K}_{K}+\frac{\ca^{2}}{\ell\rho}\nu_{K}\big)\Big)-\frac{\ca^{4}\nu^{I}}{6\rho}\nonumber \\
 & +\frac{\ca^{4}\mathcal{F}}{4}\big(\rho^{3}(\rho^{-2}\nu^{I})'\big)'+\frac{\ell\ca^{2}\mathcal{F}}{72}\Big(\rho^{4}\big(\frac{6\ca^{2}\mathcal{F}+3\mathcal{K}_{I}+4\ca^{2}-10}{\rho}\big)'\Big)''\Big]'\,,\nonumber \\
\end{align}
and
\begin{align}\label{int-aux-SU2}
\mathcal{E} & =\Bigg[\frac{\ca^{6}\mathcal{F}}{6}\Bigg(
  \frac{1}{\ell^{2}}\sum_{I=1}^{3}\Big(12(\rho^{-1}\nu_{I}^{2})'-9\nu_{I}'\phantom{}^{2}\Big)'+\frac{2}{\ell\ca^{2}}\sum_{I=1}^{3}(\mathcal{K}_{I}\nu^{I})''\nonumber \\
 & -\Big(6(\rho\mathcal{F}^{2})'-8\rho^{2}\big(\rho^{-1}(\rho\mathcal{F})'\big)'-\frac{4}{\ca^{2}}(\rho^{2}\mathcal{F})''+18\rho^{2}(\mathcal{F}\mathcal{F}''-\mathcal{F}'\phantom{}^{2})\nonumber \\
 & -\frac{9}{2}\rho^{4}\mathcal{F}''\phantom{}^{2}+\rho^{3}\big(3\mathcal{F}(3\mathcal{F}+\rho\mathcal{F}')''-\frac{21}{2}\mathcal{F}'\phantom{}^{2}\big)'\Big)'\Bigg)\Bigg]'\,,\nonumber \\
\end{align} 
where we have also used \eqref{eq:F1-SU2} and \eqref{ck-constraint} to rewrite them.

The assumption that the horizon is analytic implies that the metric in GNC is analytic in $\lambda$ and hence by Lemma \ref{Calabi-horizon} and \ref{CT-NH-gauge} that $\mathcal{F}(\rho), \nu^I(\rho)$ are analytic functions of $\rho$,
\begin{equation}
\mathcal{F}(\rho)=\ca^{-2}+\sum_{n=1}^{\infty}\cf_{n}\rho^{n}\,,\qquad\nu^{I}(\rho)=\sum_{n=2}^{\infty}\nu_{n}^{I}\rho^{n}\,,
\end{equation}
where we have taken into account \eqref{eq:FG-Aux} to fix the zeroth order term in $\mathcal{F}(\rho)$ and Lemma \ref{CT-NH-gauge} to fix the first and second order terms in $\nu^{I}(\rho)$. 
We next examine the higher order terms. 

First observe that \eqref{eq:F1-SU2} implies that
\begin{equation}\label{eq:Fn-SU2}
\sum_{I=1}^{3}\nu_{n+1}^{I}=-\ell(n+2)\cf_{n}\,,\qquad n\geq1\,.
\end{equation}
The Maxwell equations \eqref{mxwl-aux-SU2} at the first non-trivial order  yield,
\begin{equation}\label{mxwl-SU2-ord2}
\mathcal{E}^{I}=\frac{\ca^{2}}{18\rho}\Big((2-2\ca^{2}+\mathcal{K}_{I})\nu_{2}^{I}+\sum_{J,K=1}^{3}C^{IJK}(-1+\ca^{2}+\mathcal{K}_{J})\nu_{2}^{K}\Big)+O(\rho^{0})\,.
\end{equation}
This is a linear system of three equations with three unknowns $\nu_{2}^{I}$, but subject to the constraint \eqref{ck-constraint}. In fact the vanishing of \eqref{mxwl-SU2-ord2} implies $\nu_{2}^{1}=\nu_{2}^{2}=\nu_{2}^{3}$ and hence \eqref{eq:Fn-SU2} implies we have
\begin{equation}
\nu_{2}^{I}=-\ell\mathcal{F}_{1}\,.
\end{equation}
We will see that the constant $\mathcal{F}_{1}$ cannot be determined  (\eqref{int-aux-SU2} is automatically satisfied at the relevant order $\mathcal{E}=O(\rho^{0})$) and in fact it is the only free constant of the solution.

For the higher order terms, we will show by induction
\begin{equation}\label{eq:cfn}
\cf_{n}=\nu_{n+1}^{I}=0\,,\qquad n\geq2\,.
\end{equation}
The first step of the inductive argument is to verify the validity of this for $n=2$. Equation \eqref{mxwl-aux-SU2}  now yields 
\begin{align}
\mathcal{E}^{I} & =\frac{\ca^{2}}{18}\Bigg(\Big(\frac{7}{4}\mathcal{K}_{I}-2(2\ca^{2}-11)\Big)\nu_{3}^{I}-\sum_{J,K=1}^{3}\Big(C^{IJK}\big(\frac{1}{4}\mathcal{K}_{I}-(2\ca^{2}-11)-2\mathcal{K}_{J}\big)\nu_{3}^{K}\Big)\Bigg)\nonumber \\
 & +O(\rho)\,.
\end{align}
Due to \eqref{ck-constraint} the vanishing of the above equation has a one-parameter family of solutions which is most conveniently parametrised by $\cf_{2}$ through \eqref{eq:Fn-SU2}, namely\footnote{Note that \eqref{c1c2-bounds} guarantees $(11-2\ca^{2})^{2}+4\mathcal{C}_{1}\neq 0$. In particular for $\cb^{2}=\ca^{2}$ it implies $(11-2\ca^{2})^{2}+4\mathcal{C}_{1}>81+36(1-\ca^{2})>0$ since $\ca^{2}<1$.}
\begin{equation}\label{solnu-ord2}
\nu_{3}^{I}=-\frac{\ell\mathcal{F}_{2}}{3}\frac{4(11-2\ca^{2})^{2}+(11-2\ca^{2})\mathcal{K}_{I}-2\mathcal{K}_{I}^{2}+12\mathcal{C}_{1}}{(11-2\ca^{2})^{2}+4\mathcal{C}_{1}}\,.
\end{equation}
Then \eqref{int-aux-SU2} yields
\begin{equation}\label{int-aux-SU2-ord2}
\mathcal{E}=-\mathcal{F}_{2}\ca^{2}\frac{8(1-\ca^{2})(11-2\ca^{2})^{2}-12(1+2\ca^{2})\mathcal{C}_{1}-4\mathcal{C}_{2}}{(11-2\ca^{2})^{2}+4\mathcal{C}_{1}}+O(\rho)\, .
\end{equation}
Recall $\mathcal{C}_{1}$ and $\mathcal{C}_{2}$ are defined in \eqref{ccal-def}. The factor multiplying $\cF_2$ in the numerator of \eqref{int-aux-SU2-ord2} can never vanish since
\begin{align}
8(1-\ca^{2})(11-2\ca^{2})^{2}-12(1+2\ca^{2})\mathcal{C}_{1}-4\mathcal{C}_{2} & >8(1-\ca^{2})(11-2\ca^{2})^{2}-8(1-\ca^{2})^{3}\nonumber \\
 & =24(1-\ca^{2})(10-\ca^{2})(4-\ca^{2})\nonumber \\
 & >0\,,
\end{align}
where in the first line we have used \eqref{c1c2-bounds} and in the last $\ca^{2}<1$. Therefore the vanishing of \eqref{int-aux-SU2-ord2} implies that $\mathcal{F}_{2}=0$ and hence \eqref{solnu-ord2} implies $\nu_3^I=0$. Therefore, we have shown that  \eqref{eq:cfn} holds for $n=2$.

 We will now assume that \eqref{eq:cfn} holds for some $n\geq 2$ and prove that it also holds for $n+1$. As in the $n=2$ step we start from the Maxwell equations \eqref{mxwl-aux-SU2} expanded to order $O(\rho^{n-1})$,
\begin{align}\label{mxwl-SU2-ordn}
\mathcal{E}^{I} & =\frac{(n+1)\ca^{2}}{36}\Bigg(\Big(\frac{n+6}{n+3}\mathcal{K}_{I}+2a_{n}\Big)\nu_{n+2}^{I}-\sum_{J,K=1}^{3}\Big(C^{IJK}\big(\frac{n\mathcal{K}_{I}}{n+3}+a_{n}-2\mathcal{K}_{J}\big)\nu_{n+2}^{K}\Big)\Bigg)\rho^{n-1}+O(\rho^{n})\,,
\end{align}
where we have introduced for convenience
\begin{equation}
a_{n} :=3n^{2}+6n+2-2\ca^{2}\,.
\end{equation}
Taking into account \eqref{ck-constraint} and \eqref{eq:Fn-SU2} as before we find that the solution to the vanishing of \eqref{mxwl-SU2-ordn} is\footnote{Again from \eqref{c1c2-bounds} we have $a_{n}^{2}+4\mathcal{C}_{1}>3n(n+2)(3n^{2}+6n+4-4\ca^{2})>0$ for $n\geq 2$.}
\begin{equation}\label{solnu-ordn}
\nu_{n+2}^{I}=-\frac{\ell\mathcal{F}_{n+1}}{3}\frac{(n+3)a_{n}^{2}+na_{n}\mathcal{K}_{I}-2n\mathcal{K}_{I}^{2}+12\mathcal{C}_{1}}{a_{n}^{2}+4\mathcal{C}_{1}}\,,
\end{equation}
and inserting into \eqref{int-aux-SU2} we get
\begin{equation}\label{int-aux-SU2-ordn}
\mathcal{E}=-\mathcal{F}_{n+1}\frac{n^{2}(n+1)(n+2)\ca^{2}}{6}\frac{\big(a_{n}-3(1+2\ca^{2})\big)a_{n}^{2}-12(1+2\ca^{2})\mathcal{C}_{1}-4\mathcal{C}_{2}}{a_{n}^{2}+4\mathcal{C}_{1}}\rho^{n-1}+O(\rho^{n})\,.
\end{equation}
It is worth noting that \eqref{mxwl-SU2-ordn}, \eqref{solnu-ordn} and \eqref{int-aux-SU2-ordn}    reproduce respectively \eqref{mxwl-SU2-ord2}, \eqref{solnu-ord2} and \eqref{int-aux-SU2-ord2} for $n=1$. Examining the numerator in \eqref{int-aux-SU2-ordn} we have
\begin{align}
\big(a_{n}-3(1+2\ca^{2})\big)a_{n}^{2}-12(1+2\ca^{2})\mathcal{C}_{1}-4\mathcal{C}_{2} & >\big(a_{n}-3(1+2\ca^{2})\big)a_{n}^{2}-8(1-\ca^{2})^{3}\nonumber \\
 & >\big(a_{1}-3(1+2\ca^{2})\big)a_{1}^{2}-8(1-\ca^{2})^{3}\nonumber \\
 & =24(1-\ca^{2})(10-\ca^{2})(4-\ca^{2})\nonumber \\
 & >0\,,
\end{align}
where we have used \eqref{c1c2-bounds} as well as the fact that $a_{n} > a_{1}>3(1+2\ca^{2})$ for $n\geq2$. Therefore, the  vanishing of  \eqref{int-aux-SU2-ordn} implies that $\mathcal{F}_{n+1}=0$ and hence \eqref{solnu-ordn} implies $\nu_{n+2}^I=0$.  We have therefore shown that  \eqref{eq:cfn} also holds for $n+1$.

Therefore, by induction, it follows that \eqref{eq:cfn} is true. Redefining the integration constant as
\begin{equation}
\mathcal{F}_{1}=\ca^{-2}F_{3}\,,
\end{equation}
we deduce that the only analytic solution to \eqref{eq:isotropic}, \eqref{mxwl-CT} and \eqref{int-CT} in the case $\ca^{2}=\cb^{2}$ is \eqref{sol-KLR}.
\end{proof}

It remains to find the rest of the supersymmetric data in the timelike decomposition \eqref{metricform} and \eqref{maxwellform}, which is given by the following.

\begin{lemma}\label{lem:rest-sol} Given a supersymmetric solution as in Proposition \ref{prop:base}, the function $f$ is 
\begin{equation}
f=\frac{12\rho}{\ell^{2}}\Bigg(\prod_{I=1}^{3}\Big(\big(3s\rho+\Delta_{2}(\eta)\big)\ell\fil_{I}+\ck_{I}\Big)\Bigg)^{-1/3}\, ,
\end{equation}
and the scalar fields are
\begin{equation}\label{eq:stu_XI-Calabi}
    X^I = \dfrac{12 \rho}{f \ell^2 } \left( \big( 3s\rho\,+\, \Delta_{2}(\eta) \big)\,\ell\,\zeta_I \,+\, \mathcal{K}_I  \right)^{-1}\,.
\end{equation}
The axis set is $\eta=\pm 1$ and corresponds to the fixed points of $\partial_{\phi^1}$ and $\partial_{\phi^2}$ respectively. The 1-form $\omega$ is 
\begin{align}
\omega  =\Big(\frac{\ell^{3}s^{2}}{8}\rho+\frac{\ell^{3}s}{8}(1-\Delta_{1}(\eta))-\frac{\ell^{3}\Delta_{3}(\eta)}{48\rho}\Big)\sigma-\frac{\ell^{3}(\ca^{2}-\cb^{2})}{16}\Big(\frac{s}{\Delta_{1}(\eta)}-\frac{1}{\rho}\Big)\tau  \; ,
\end{align}
and the magnetic part of the  gauge fields  are (up to a gauge transformation),
\begin{align}
    A^I &= -\,\frac{\ell}{2} s \rho \bar{X}^I \sigma \,  \nonumber \\
                &- \frac{\ell^2}{6} C^{IJK}\zeta_J \left( \frac{3}{4} \left(\mathcal{A}^2 - \mathcal{B}^2 \right) (1- \eta^2)\,\ell \zeta_K \,+\, \left( (1- \frac{1}{2} (\mathcal{A}^2+\mathcal{B}^2)) \,\ell \zeta_K +\mathcal{K}_K \right)\,\eta \right)  \mathrm{d}\varphi\, . \label{eq:stu_A-Calabi}
\end{align}
The functions $\Delta_i(\eta)$ are given by \eqref{Delta-functions}. 

\end{lemma}

\begin{proof}  First from the explicit form of the  scalars \eqref{eq:PhiIsol} and their definition \eqref{eq_Phidef} we can invert  to obtain the original scalars,
\be
f^{-1} X_I= \frac{\ell^2}{36} \frac{(3 \rho s +\Delta_2(\eta)) \ell \zeta_K + \ck_K }{\rho} \; ,  \label{eq:fXI} 
\ee
and hence using $C^{IJK} X_I X_J X_K=2/9$ we obtain the claimed form for $f$ and $X^I$.  In particular, for $\rho>0$ the numerator of \eqref{eq:fXI} is strictly positive (since it is for the near-horizon geometry), so  we deduce that away from the horizon $f>0$. It therefore follows from \eqref{eq_invariants} that on the axis set $G_{ij}$ does not have full rank, so $\det G_{ij}=0$.  On the other hand, from the explicit form of the K\"ahler base $\det G_{ij}= N F(\rho) G(\eta)$ and the functions $F(\rho), G(\eta)$, we see that the only way this can vanish for $\rho>0$ is if $\eta=\pm 1$.

We now turn to the 1-form $\omega$. Using the data in Proposition \ref{prop:base}, we find that
\eqref{lambda-CT} give
\begin{align}
\partial_{\rho}\Big(\sqrt{FG}\lambda_{3}\Big) & =\frac{3(\ca^{2}-\cb^{2})G}{\rho^{2}}\,,\nonumber \\
\partial_{\eta}\Big(\sqrt{FG}\lambda_{3}\Big) & =2\Big(\frac{1}{\rho}+s\Big)\Big(\Delta_{2}^{2}+\Delta_{3}+\cc_{1}\Big)\,,
\end{align}
which imply the solution is
\begin{equation}
\sqrt{FG}\lambda_{3}=-3(\ca^{2}-\cb^{2})\frac{FG}{\rho^{3}}+\lambda_{3,0}\,,   \label{eq:lambda3}
\end{equation}
where $\lambda_{3,0}$ a constant.  On the other hand, the equation for $\omega$ \eqref{dom}, reduces to the following PDEs for $\omega_{\psi}$ and $\omega_{\varphi}$, 
\begin{align}
-\frac{48}{\ell^{3}}\partial_{\rho}\omega_{\psi} & =-\Big(\frac{\Delta_{3}}{\rho^{2}}+6s^{2}\Big)\,,\nonumber \\
-\frac{48}{\ell^{3}}\partial_{\rho}\omega_{\varphi} & =-\eta\Big(\frac{\Delta_{3}}{\rho^{2}}+6s^{2}\Big)+\lambda_{3}\rho\sqrt{\frac{G}{F}}\,,\nonumber \\
-\frac{48}{\ell^{3}}\partial_{\eta}\omega_{\psi} & =-\lambda_{3}\sqrt{\frac{F}{G}}\,,\nonumber \\
-\frac{48}{\ell^{3}}\partial_{\eta}\omega_{\varphi} & =-\Big(6s(\Delta_{2}+s\rho)+\frac{2\Delta_{2}^{2}+\Delta_{3}+2\cc_{1}}{\rho}\Big)-\lambda_{3}\eta\sqrt{\frac{F}{G}\, .}
\end{align}
Eliminating $\lambda_3$ using \eqref{eq:lambda3} and integrating one finds the solution is
\begin{align}
\omega & =\Big(\frac{\ell^{3}s^{2}}{8}\rho+\frac{\ell^{3}s}{8}(1-\Delta_{1}(\eta))-\frac{\ell^{3}\Delta_{3}(\eta)}{48\rho}\Big)\sigma-\frac{\ell^{3}(\ca^{2}-\cb^{2})}{16}\Big(\frac{s}{\Delta_{1}(\eta)}-\frac{1}{\rho}\Big)\tau+\omega_{0}\nonumber \\
 & -\frac{\ell^{3}\lambda_{3,0}}{48}\Big\{\frac{\log(s+\rho^{-1})}{(1-\eta^{2})\Delta_{1}(\eta)}\tau+\frac{1}{\ca^{2}\cb^{2}}\Big[\log\Big(\frac{2\Delta_{1}(\eta)}{1+\eta}\Big)\td\phi^{1}-\log\Big(\frac{2\Delta_{1}(\eta)}{1-\eta}\Big)\td\phi^{2}\Big]\Big\}\, ,
\end{align}
where $\omega_0=\omega_{0i} \td \phi^i$ and $\omega_{0i}$ are constants. Now, imposing that the spacetime metric is smooth at the horizon implies that $\omega_i$ near the horizon behaves as \eqref{omNH} in GNC and hence, by the coordinate change \eqref{rhoeta-NH}, must be smooth function of $\rho$ apart from a $1/\rho$ leading pole.  This implies that the constant $\lambda_{3,0}=0$. Furthermore, imposing that $\omega$ is a smooth 1-form at the axis $\eta=\pm 1$ implies that  $\omega_0=0$, giving the claimed form. 

Finally, the gauge field follows from \eqref{eq:CT-gaugefield}, together with \eqref{muI-eta}, \eqref{muIrho} and \eqref{eq:psiphitransf} where we have fixed a gauge for $A^I$.
\end{proof}

We have  completely determined the solution under our assumptions, which is given by Proposition \ref{prop:base} and Lemma \ref{lem:rest-sol}.  We now show that this is locally isometric to the known  black hole solution or its near-horizon geometry.  We have provided a simplified form for this black hole solution in Appendix \ref{sec:KLR} which is convenient for comparison to our general solution.  It is straightforward to check that the solution with $s=4/\ell^2$ and $s=0$ are identical to the known black hole and its near-horizon geometry respectively, upon the coordinate change
\be
\rho = r^2/4, \qquad \eta = \cos\vartheta \; ,
\ee
and the parameter identification (recall also that  $\fil_I:= 3 \ell^{-1} \bar{X}_I$)
\be
\ca^2= A^2, \qquad \cb^2= B^2\,  \qquad  \mathcal{K}_I = K_I\,.
\ee
This completes the proof of Theorem \ref{stu-theorem}.

Finally, we note that the K\"ahler base in Proposition  \ref{prop:base} has an enhanced local $SU(2)\times U(1)$ symmetry if and only if $\ca^2=\cb^2$. Furthermore, in this case the separable magnetic 2-forms $F^I$, which have gauge fields \eqref{eq:CT-gaugefield}, also possess a local $SU(2)\times U(1)$ symmetry. Therefore, the uniqueness theorem for the case $\ca^2=\cb^2$  also establishes the uniqueness of supersymmetric solutions with a local $SU(2)\times U(1)$ symmetry  that are timelike outside a compact horizon. This is because timelike supersymmetric solutions that possess a symmetry which leaves $V$ invariant,  must have a  K\"ahler base and magnetic fields $F^I$ that are also invariant under the symmetry (this can be argued in essentially the same way as in minimal supergravity~\cite{Lucietti:2021bbh}). Furthermore, a K\"ahler metric with a local $SU(2)\times U(1)$ symmetry is a special case of Calabi-toric (with $G(\eta)\propto (1-\eta^2)$). This completes the proof of Theorem \ref{stu-su2-heorem}.

\section{Discussion}
\label{sec:discussion}

In this paper we have proven a uniqueness theorem for the most general known supersymmetric black hole solution  in five-dimensional STU gauged supergravity~\cite{Kunduri:2006ek}.  The key assumption in our theorem is  a  toric symmetry that is compatible with supersymmetry and that is separable in a sense that we defined.  In particular, the concept of separability implies that the solution is specified by single-variable functions in an orthogonal coordinate system and hence allows one to reduce the problem to ODEs.  We find that for solutions containing horizons, with compact locally spherical cross-sections, the near-horizon boundary conditions fix the angular dependent functions while the radial dependent functions are determined by solving the remaining ODEs. Therefore, our proof is constructive, and furthermore, results in a simpler form of the solution. This generalises our previous uniqueness theorem in minimal gauged supergravity~\cite{Lucietti:2022fqj} in two directions: to the more general STU theory and to the broader class of separable K\"ahler bases.   Our work leaves a number of open problems which we now elaborate on.

The near-horizon classification of toric supersymmetric solutions also contains solutions with non-spherical horizon topology, namely horizons with cross-sections of topology $S^1\times S^2$ and $T^3$~\cite{Kunduri:2007qy}. These solutions are not allowed in minimal gauged supergravity, however, they are in certain regions of the scalar moduli space in the STU theory.    There are also near-horizon geometries with horizon cross-section topology $S^1$ times a  2d spindle~\cite{Ferrero:2021etw} (these are allowed in minimal supergravity~\cite{Ferrero:2020laf}). In fact all of these near-horizon geometries are null supersymmetric solutions and therefore not covered by our analysis.  In order to extend our work to these cases would require performing a higher order calculation in order to determine the leading near-horizon behaviour of the K\"ahler base for a timelike supersymmetric solution containing such a horizon.  This is an interesting open problem that would in particular clarify the existence of possible supersymmetric black rings, strings and spindles in this theory. 

It would also  be interesting to investigate the assumption  of separability further, in particular, whether any progress can be made by relaxing this assumption and study generic supersymmetric toric solutions. The main motivation for studying this class is that it can accommodate topologically non-trivial spacetimes and horizons and hence this is the natural symmetry class within which to address the existence of black rings,  black lenses, black holes in bubbling spacetimes or even multi-black holes. This appears to be a very complicated problem in toric K\"ahler geometry, even in the case of minimal supergravity, which ultimately may require the use of numerical techniques. On the geometrical side, we  showed that  separability of a toric K\"ahler metric is equivalent to the existence of a Hamiltonian 2-form. It would be interesting to clarify the spacetime interpretation of this structure and whether it is related to other notions of separability such as the existence of Killing-Yano tensors or related structures (which are known to exist for the black hole  in minimal supergravity~\cite{Kubiznak:2009qi}).

This theory also admits supersymmetric solitons, that is, finite energy spacetimes that are everywhere smooth with no event horizons~\cite{Cassani:2015upa, Durgut:2021rma, Lucietti:2021bbh, Durgut:2023rmu}.  An essential part of our analysis was to assume the existence of a smooth horizon and hence our results do not include such solutions. It would be interesting if a uniqueness theorem for such solitons could be established.  In fact, the known supersymmetric soliton that is asymptotically globally AdS$_5$ has an orthotoric K\"ahler base~\cite{Cassani:2015upa, Durgut:2021rma}, whereas the asymptotically locally AdS$_5$ solitons have an $SU(2)\times U(1)$ invariant base and hence are Calabi-toric~\cite{Lucietti:2021bbh, Durgut:2023rmu}. Thus all the known solitons are separable according to our definition.   It is therefore  plausible that one may be able to establish a classification theorem for solitons for separable supersymmetric toric solutions. 

The known supersymmetric black hole~\cite{Kunduri:2006ek} is expected to arise as the BPS limit of the non-extremal three-charged rotating black hole in STU supergravity~\cite{Wu:2011gq}. The latter is a 6-parameter family that correspond to the mass, three charges and two angular momentum parameters, whereas the former is a 4-parameter family with the mass fixed by the BPS relation and a non-linear constraint between the charges and angular momenta.\footnote{In the special case with $SU(2)\times U(1)$ spacetime symmetry (which posses equal angular momenta) the non-extremal solutions in the STU theory were found in \cite{Cvetic:2004ny} and their  BPS limit was studied in \cite{Cvetic:2005zi} (see also \cite{Ntokos:2021duk}).}  We expect there is a 5-parameter family of supersymmetric solutions that correspond to the BPS limit of the non-extremal black hole~\cite{Wu:2011gq} which generically do not have a black hole interpretation, but can be analytically continued to obtain smooth `complex saddles' that should be relevant in holography.  In the special case of minimal gauged supergravity (which possess three equal charges) the non-extremal black hole is the CCLP solution~\cite{Chong:2005hr} and the BPS complex saddles where first found in~\cite{Cabo-Bizet:2018ehj} and possess an orthotoric K\"ahler base~\cite{Cassani:2015upa}. We expect the aforementioned more general family of complex saddles in the STU theory to also be a separable supersymmetric solution with an orthotoric K\"ahler base.  This may also allow one to write the BPS limit of the solution~\cite{Wu:2011gq} in a simplified form.  \\

\noindent {\bf Acknowledgements.}   JL and PN acknowledge support by the Leverhulme Research Project Grant RPG-2019-355.  SO acknowledges support by a Principal's Career Development Scholarship at the University of Edinburgh.

\appendix

\section{Exceptional separable toric K\"ahler metrics}\label{sec:separability-appendix}

In this Appendix we will consider  non-separable Hamiltonian 2-forms on separable K\"ahler surfaces. As shown in Section \ref{sec:ham-2-forms} this corresponds to the case where the constant $p$ is non-vanishing in \eqref{eq:mu-aux}. We show below this is possible only when the functions $F(\xi)$ and $G(\eta)$ in \eqref{eq:metric-separable} have the form in Table \ref{table:exc-sep-ham}. We also provide the moment maps for the corresponding Hamiltonian 2-forms.  Observe that these examples all have multiple non-trivial Hamiltonian 2-forms (see also~\cite{Apostolov2001TheGO}).

\begin{table}[h!]
\centering
\begin{tabular}{|c|c|c|c|c|}
\hline 
Class & $F(\xi)$ & $G(\eta)$ & $\hmu_{\xi}(\xi)$ & $\hmu_{\eta}(\eta)$\tabularnewline
\hline 
\hline 
PT & $F_{0}+F_{1}\xi$ & $G_{0}+G_{1}\eta$ & $pG_{1}\xi^{2}+\gamma_{1}\xi$ & $pF_{1}\eta^{2}+\delta_{1}\eta$\tabularnewline
\hline 
CT & $F_{2}\xi^{2}+F_{3}\xi^{3}$ & $-F_{2}\eta^{2}+G_{1}\eta+G_{0}$ & $\gamma_{2}\xi^{2}+\gamma_{3}\xi^{3}$ & $0$\tabularnewline
\hline 
OT & $F_{0}+F_{1}\xi+F_{2}\xi^{2}+F_{3}\xi^{3}$ & $-(F_{0}+F_{1}\eta+F_{2}\eta^{2}+F_{3}\eta^{3})$ & $\gamma_{2}\xi^{2}+\gamma_{3}\xi^{3}$ & $-(\gamma_{2}\eta^{2}+\gamma_{3}\eta^{3})$\tabularnewline
\hline 
\end{tabular}
\caption{For exceptional separable toric K\"ahler surfaces the Hamiltonian 2-form is not necessarily separable. $F_{n}, G_{n}, \gamma_{n}$ and $\delta_{n}$ are constants.}
\label{table:exc-sep-ham}
\end{table}

Starting with the PT case, we have
\begin{equation}
0=\partial_{\xi}E_{1}=p\big(2F''(\xi)-G''(\eta)\big)\,,\qquad0=\partial_{\eta}E_{2}=p\big(2G''(\eta)-F''(\xi)\big)\,,
\end{equation}
which implies $F(\xi)$ and $G(\eta)$ are linear as shown in the first row of 
Table \ref{table:exc-sep-ham}. Solving  $0=E_{1}=E_{2}$ we find $\hmu_{\xi}(\xi)=pG_{1}\xi^{2}+\gamma_{1}\xi+\gamma_{0}$ and
$\hmu_{\eta}(\eta)=pF_{1}\eta^{2}+\delta_{1}\eta+\delta_{0}$ where
the constant terms can be fixed to zero using the gauge transformations \eqref{eq:gauge-mu-xieta}. The results are summarised
in the first row of Table \ref{table:exc-sep-ham}.

Next for the CT case we have
\begin{equation}
0=\partial_{\xi}\partial_{\eta}E_{1}=-pG'''(\eta)\,,\qquad0=\partial_{\xi}^{2}\partial_{\eta}(\xi^{2}E_{2})=-p\xi^{2}F''''(\xi)\,,
\end{equation}
which imply $G(\eta)=G_{0}+G_{1}\eta+G_{2}\eta^{2}$ and $F(\xi)=F_{0}+F_{1}\xi+F_{2}\xi^{2}+F_{3}\xi^{3}$.
Further using $0=\partial_{\xi}E_{1}$ we get $F_{0}=0$ and $G_{2}=-F_{2}.$
We then have $0=E_{1}=-4pF_{1}-\tilde{\mu}_{\eta}''(\eta)$ and hence
$\tilde{\mu}_{\eta}(\eta)=-2pF_{1}\eta^{2}+\delta_{1}\eta+\delta_{0}$
and inserting into $0=\partial_{\eta}(\xi^{2}E_{2})$ we find $F_{1}=0$.
With these $0=E_{2}$ becomes and ODE for $\hmu_{\xi}(\xi)$ which
can be readily solved to give $\hmu_{\xi}(\xi)=\gamma_{3}\xi^{3}+\gamma_{2}\xi^{2}-\delta_{1}\xi$.
We can then use gauge transformations  \eqref{eq:gauge-mu-xieta} to
get the results in the second row of Table \ref{table:exc-sep-ham}.

For OT geometries we have 
\begin{equation}
0=\partial_{\xi}\Big[\frac{1}{(\xi-\eta)^{2}}\partial_{\eta}^{2}\big((\xi-\eta)^{2}E_{1}\big)\Big]=-pG''''(\eta)\,,\qquad0=\partial_{\eta}\Big[\frac{1}{(\xi-\eta)^{2}}\partial_{\xi}^{2}\big((\xi-\eta)^{2}E_{2}\big)\Big]=-pF''''(\xi)\,,
\end{equation}
from which we infer that $F(\xi)$ and $G(\eta)$ are cubic polynomials.
We then have $0=\partial_{\xi}^{2}\partial_{\eta}^{2}E_{1}=-2\tilde{\mu}_{\eta}''''(\eta)$
and $0=\partial_{\xi}^{2}\partial_{\eta}^{2}E_{2}=-2\tilde{\mu}_{\xi}''''(\xi)$
and therefore $\hmu_{\xi}(\xi)$ and $\tilde{\mu}_{\eta}(\eta)$ are
cubic polynomials as well. Then $0=E_{1}=E_{2}$ are polynomial equations
in $\xi$ and $\eta$ and we can easily deduce that $F(\xi)$ and
$G(\eta)$ should have opposite coefficients and the same holds for
$\hmu_{\xi}(\xi)$ and $\tilde{\mu}_{\eta}(\eta)$. Finally using
gauge transformations  \eqref{eq:gauge-mu-xieta} we can arrive at the third row of Table \ref{table:exc-sep-ham}.

\section{The known black hole}\label{sec:KLR}

The known supersymmetric black hole in STU gauged supergravity is a four parameter family of supersymmetric solutions~\cite{Kunduri:2006ek}. We present it here in a  simplified form, also providing the necessary formulae to compare  to our notation.

The parameters of the solution are  $0<A^{2},B^{2}<1$ and $K_{I}$, $I=1,2,3$  subject to 
\begin{equation}
C^{IJK}\bar{X}_{I}\bar{X}_{J}K_{K}=0\,,
\end{equation}
where $\bar{X}_{I}$ were introduced in subsection \ref{sec:time-class}. The parameter space of the solution is further constrained by 
\begin{equation}
\kappa^{2}(A^{2},B^{2},C_{1},C_{2})> 0\,,\qquad K_{I}>\frac{3}{2}(A^{2}+B^{2})+\frac{1}{2}|A^{2}-B^{2}|-1\,,
\end{equation}
where
\begin{equation}\label{ccal-def-2}
C_{1}=\frac{\ell}{6}C^{IJK}\bar{X}_{I}K_{J}K_{K}\,,\qquad C_{2}=\frac{1}{6}C^{IJK}K_{I}K_{J}K_{K}\, ,
\end{equation}
and $\kappa^{2}$ is given by \eqref{kappa-def}. The metric and the K\"ahler form of the K\"ahler base are given by
\begin{align}\label{eq:stu_KLRmetr}
 h & =\frac{\mathrm{d}  r^{2}}{V( r)}+\frac{ r^{2}}{4}\left(\frac{\mathrm{d}\vartheta^{2}}{\Delta_{\vartheta}}+\Delta_{\vartheta}\sin^{2}\vartheta\mathrm{d}\phi^{2}\right)+\frac{ r^{2}V( r)}{4}(\mathrm{d}\psi+\cos\vartheta\mathrm{d}\phi)^{2}\,,\nonumber \\
    J & =\mathrm{d}\left(\tfrac{1}{4} r^{2}(\mathrm{d}\psi+\cos\vartheta\mathrm{d}\phi)\right)\,,
\end{align}
where $\Delta_{\vartheta}=A^{2}\cos^{2}(\vartheta/2)+B^{2}\sin^{2}(\vartheta/2)$. The  black hole corresponds to $V=1+\frac{ r^2}{\ell^2}$ while its near-horizon geometry (also a supersymmetric solution)  corresponds to $V=1$. The coordinate ranges are $r\geq 0$ and $0 \leq \vartheta \leq \pi$ while the angles $\psi$ and $\varphi$ are given in terms of $2\pi$-periodic coordinates $\phi^{i}$ as
\begin{equation}\label{eq:stu_KLRangles}
\psi=A^{-2}\phi^{1}+B^{-2}\phi^{2}\,,\qquad\phi=-A^{-2}\phi^{1}+B^{-2}\phi^{2}\,.
\end{equation}
An interesting observation about \eqref{eq:stu_KLRmetr} is the fact that it does not involve the constants $K_{I}$. In particular the K\"ahler base is the same for $K_{1}=K_{2}=K_{3}=0$ which yields the supersymmetric CCLP solution \cite{Chong:2005hr} (see appendix B of \cite{Lucietti:2022fqj}).

In order to write the full solution it is convenient to introduce $\tilde{r}^{2}$ (note that it is $\vartheta$-dependent) through
\begin{equation}
\tilde{r}^{2}/\ell^{2}=V(r)-\Delta_{\vartheta}+\frac{1}{3}(A^{2}+B^{2}-2)\,.
\end{equation}
With this we have 
\begin{equation}
f^{-3}=\prod_{I=1}^{3}\Big(\frac{3\tilde{r}^{2}}{r^{2}}\bar{X}_{I}+\frac{\ell^{2}}{3r^{2}}K_{I}\Big)\,,
\end{equation}
and
\begin{align}
\omega  &= \left(\frac{ r^{2}}{2\ell}+\frac{\ell}{2}(1-\Delta_{\vartheta})-\frac{\ell^{3}}{12 r^{2}}\big(6\Delta_{\vartheta}-A^{4}-B^{4}+A^{2}B^{2}-2(A^{2}+B^{2})-1 -C_1\big)\right)(\mathrm{d}\psi+\cos\vartheta\mathrm{d}\phi)\nonumber \\
        &+ \frac{\ell(A^{2}-B^{2})}{4} \left( 1\,-\, \frac{\ell^{2}\Delta_{\vartheta}}{ r^{2}}  \right) \sin^{2}\vartheta\,\mathrm{d}\phi\,.
\end{align}
Finally, the scalar and gauge fields are respectively given by
\begin{equation}
f^{-1}X_{I}=\frac{\tilde{r}^{2}}{r^{2}}\bar{X}_{I}+\frac{\ell^{2}}{9r^{2}}K_{I}\,,
\end{equation}
and
\begin{align}
     \mathbf{A}^{I} &= X^I f (\mathrm{d} t + \omega) \,-\,\epsilon\, \frac{r^2}{2 \ell} \bar{X}^I \,(\mathrm{d}\psi+\cos\vartheta\mathrm{d}\phi)\, \nonumber\\
                &- \frac{\ell}{2} C^{IJK} \bar{X}_J \left( \frac{9}{4}   \left(A^2 - B^2 \right)  \bar{X}_K\sin^2 \vartheta \,+\, \big(3 \bar{X}_K- \frac{3}{2} (A^2+B^2) \bar{X}_K + K_{K}\big)\,\cos \vartheta \right) \mathrm{d}\phi\, ,
\end{align}
where $\epsilon = 1$ for the black hole and $\epsilon = 0$ for its near-horizon geometry.

The map between our coordinates and parameters and those in \cite{Kunduri:2006ek} is given by the following,
\begin{equation}
\ell_{\textrm{here}}=g_{\textrm{there}}^{-1},\quad( r/\ell)_{\textrm{here}}=\sinh(g\sigma )_{\textrm{there}}\,,\quad(\phi^{1},\phi^{2},\vartheta)_{\textrm{here}}=(-\phi,-\psi,2\theta)_{\textrm{there}}\,,
\end{equation}
\begin{equation}
(A^{2},B^{2},K_{I})_{\textrm{here}}=(A^{2},B^{2},9e_{I}/\ell^{2})_{\textrm{there}}\,,\quad\Delta_{\vartheta}|_{\textrm{here}}=\frac{\Delta_{\theta}}{g^{2}\alpha^{2}}\Big|_{\textrm{there}\,}\,,  \; \quad \tilde{r}^2|_{\text{here}} = \rho^2|_{\text{there}}.
\end{equation}
We also note that for $A^2=B^2$ we recover the $SU(2)\times U(1)$-symmetric black hole found in \cite{Gutowski:2004yv}.

\section{Unified form of near-horizon geometry}\label{NH-comparison}

The possible near-horizon geometries with compact cross-sections in five-dimensional STU gauged supergravity that admit a toric symmetry were derived in~\cite{Kunduri:2007qy}. For the solutions with locally $S^3$ horizons, the cases with generic toric symmetry and enhanced $SU(2)\times U(1)$ symmetry were treated separately in that reference. The latter with enhanced symmetry first appeared in~\cite{Gutowski:2004yv}. Here we show that they can be written in a unified coordinate system as in subsection \ref{sec:NH}.

\subsection{Generic toric symmetry}

The near-horizon geometry in the case of generic toric symmetry was expressed in~\cite{Kunduri:2007qy} in terms of six parameters $x_{1},x_{2},x_{3}$ and $\knh_{1},\knh_{2},\knh_{3}$, the latter being subject to 
\begin{equation}
\knh_{1}+\knh_{2}+\knh_{3}=0\,.
\end{equation}
Let us present the near-horizon solution in this parametrisation where $\chi^{i}_{\text{here}}=x^{i}_{\text{there}}$ and $k_{I}|_{\text{here}}=K_{I}|_{\text{there}}$. The leading order of the near-horizon geometry in \eqref{NH-metric-main} is given by
\begin{align}\label{NHgeom2007}
\Delta^{(0)}= & \frac{\Delta_{0}}{H(x)^{2/3}}\,,\nonumber \\
h^{(0)}= & \Big(C^{2}-\frac{\Delta_{0}^{2}}{H(x)}\Big)\td\chi^{1}+\frac{\Delta_{0}(\alpha_{0}-x)}{H(x)}\td\chi^{2}-\frac{H'(x)}{3H(x)}\td x\,,\nonumber \\
\gamma^{(0)}= & \frac{\ell^{2}H(x)^{1/3}}{4P(x)}\td x^{2}+\frac{C^{2}H(x)-\Delta_{0}^{2}}{H(x)^{2/3}}\Big(\td\chi^{1}+\frac{\Delta_{0}(\alpha_{0}-x)}{C^{2}H(x)-\Delta_{0}^{2}}\td\chi^{2}\Big)^{2}\nonumber \\
+ & \frac{4H(x)^{1/3}P(x)}{C^{2}H(x)-\Delta_{0}^{2}}(\td\chi^{2})^{2}\,,
\end{align}
while of the near-horizon gauge field \eqref{NH-gauge-main} is given by
\begin{equation}\label{NHgauge2007}
a_{(0)}^{I}=\frac{X_{(0)}^{I}}{H(x)^{1/3}}\big(\Delta_{0}\td\chi^{1}+(x-\alpha_{0})\td\chi^{2}\big)\,,
\end{equation}
and of the scalars by
\begin{equation}\label{NHscalar2007}
X_{I}^{(0)}=\frac{1}{H(x)^{1/3}}\Big(\frac{\ell}{3}\fil_{I}x+\knh_{I}\Big)\,.
\end{equation}
The constants $\Delta_{0},C$ and $\alpha_{0}$ are given in terms of $x_{1,2,3}$ as
\bea
C^2= \frac{4}{\ell^2}(x_1+x_2+x_3),  \qquad \alpha_0= \frac{x_1x_2+ x_1x_3+x_2x_3-3c_1}{2(x_1+x_2+x_3)}, \nonumber  \\
\Delta_0^2=\frac{4 (x_1 x_2 x_3+c_2) (x_1+x_2+x_3)}{\ell^2}-\frac{(x_1x_2+ x_1x_3+x_2x_3-3c_1)^2}{\ell^2} \; ,
\eea
and the functions $H(x)$ and $P(x)$ are defined by
\begin{equation}
H(x)=\prod_{I=1}^{3}(x+3\knh_{I})=x^{3}+3c_{1}x+c_{2}\,,\quad P(x)=\prod_{i=1}^{3}(x-x_{i})=H(x)-\frac{\ell^{2}C^{2}}{4}(x-\alpha_{0})^{2}-\frac{\Delta_{0}^{2}}{C^{2}}\,,
\end{equation}
where
\begin{equation}
c_{1}=\frac{3\ell}{2}C^{IJK}\fil_{I}\knh_{J}\knh_{K}\,,\qquad c_{2}=\frac{9}{2}C^{IJK}\knh_{I}\knh_{J}\knh_{K}\,.
\end{equation}
One can easily see that the near-horizon solution is invariant under the rescalings
\begin{equation}\label{resc-NH}
x\to\tilde{K}x\,,\qquad\chi^{1}\to\tilde{K}^{-1}\chi^{1},\qquad\chi^{2}\to\chi^{2}\,,\qquad \knh_{1,2,3}\to\tilde{K}\knh_{1,2,3}\,,\qquad x_{1,2,3}\to\tilde{K}x_{1,2,3}\,,
\end{equation}
where $\tilde{K}>0$ is a constant.

The values of $x_{1},x_{2},x_{3}$ are restricted by
\begin{equation}\label{x-constraints}
0<x_{1}<x_{2}<x_{3}\,,\qquad\Delta_{0}^{2}(x_{1},x_{2},x_{3})>0\,,
\end{equation}
and for the coordinate $x$ we  have $x_{1} \leq  x \leq x_{2}$  with $P(x)>0$ in the interior. The parameters $\knh_{I}$ are also constrained by demanding positivity of the scalars on the horizon, so \eqref{NHscalar2007} gives
\begin{equation}\label{sc-pos-x}
\frac{1}{3}x_{1}+\knh_{I}>0\,.
\end{equation}

The metric generically has conical singularities at the endpoints $x= x_{1}, x_{2}$ where two different linear combinations of the biaxial Killing fields vanish~\footnote{Note a typo in \cite{Kunduri:2007qy}, equation (122)  in that paper for $\omega(x_i)$ has a missing factor of $\Delta_0$ in the numerator.}. The Killing fields $\partial_{\chi^{i}}$ do not necessarily have closed orbits and are related to the Killing fields with fixed points  $m_{i}=\partial_{\hat{\phi}^{i}}$  by
\begin{equation}\label{m-chi}
m_{i}=-d_{i}\Big(\frac{4}{\ell^{2}C^{4}(\alpha_{0}-x_{i})}\partial_{\chi^{1}}-\partial_{\chi^{2}}\Big)\,,\qquad i=1,2 ,
\end{equation}
where $m_1=0$ at $x=x_{1}$ and $m_2=0$ at $x=x_{2}$. The Killing fields $m_i$ must have closed orbits in order to avoid conical singularities.
We can determine the constants $d_{i}$ (up to signs) by demanding that $\hat{\phi}^{i}\sim\hat{\phi}^{i}+2\pi$ and that the metric has no conical singularities at these endpoints. We find 
\begin{equation}\label{di}
d_{1}=\frac{\ell^{3}C^{2}}{4(x_{2}-x_{1})}\frac{\alpha_{0}-x_{1}}{x_{3}-x_{1}}\,,\qquad 
d_{2}=-\frac{\ell^{3}C^{2}}{4(x_{2}-x_{1})}\frac{\alpha_{0}-x_{2}}{x_{3}-x_{2}}\,,
\end{equation}
where the signs have been conveniently chosen. From equations \eqref{m-chi} and \eqref{di} we can read off the matrix $A$ which determines the transformation to the $2\pi$-periodic angles $\chi^{i}=\hat{\phi}^{j}A_{j}^{\; i}$ where $m_{i} =A_{i}^{\; j} \partial_{\chi^j}$.

It's also useful to note that from \cite{Kunduri:2007qy} we can infer that 
\begin{equation}\label{NHZ02007}
Z^{(0)}=\frac{3\ell}{2\sum_{I}X^{I}}\Big[-\frac{x^{2}+c_{1}}{H}\Big(C^{2}(x-\alpha_{0})\td\chi^{1}-\frac{4\Delta_{0}}{\ell^{2}C^{2}}\td\chi^{2}\Big)+\frac{\td H}{3H}\Big]\,,
\end{equation}
where $Z^{(0)}$ is the leading order of the 1-form $Z$ appearing in \eqref{eq_JGNC}.

The near-horizon solution can be expressed in terms of quantities invariant under \eqref{resc-NH}. For this purpose, we define a new coordinate 
\begin{equation}
\hat{\eta}:=-\frac{x-x_{1}+x-x_{2}}{x_{2}-x_{1}}\,,
\end{equation}
and new parameters
\begin{equation}
\mathcal{A}^{2} :=\frac{x_{3}-x_{1}}{x_{1}+x_{2}+x_{3}}\,,\qquad\mathcal{B}^{2} :=\frac{x_{3}-x_{2}}{x_{1}+x_{2}+x_{3}}\,,\qquad\mathcal{K}_{I} :=\frac{9\knh_{I}}{x_{1}+x_{2}+x_{3}}\,,
\end{equation}
and the parameters $\cc_{1}$ and $\cc_{2}$ are defined in \eqref{ccal-def}~\footnote{Note that $
\cc_{1}=\frac{9c_{1}}{x_{1}+x_{2}+x_{3}}$, $\cc_{2}=\frac{27c_{2}}{x_{1}+x_{2}+x_{3}}$.}.
Therefore $-1\leq \hat{\eta}\leq 1$ and from the first equation in \eqref{x-constraints} we deduce that $0<\cb^{2}<\ca^{2}<1$.  We can now straightforwardly verify that the expressions \eqref{NHgeom2007}, \eqref{NHgauge2007}, \eqref{NHscalar2007} and \eqref{NHZ02007} map to \eqref{NH-expl1}, \eqref{NH-expl2},  \eqref{NH-expl3} and \eqref{Z0} respectively (up to gauge transformations for the gauge field), the second constraint in \eqref{x-constraints} maps to \eqref{kappa2} and \eqref{sc-pos-x} maps to \eqref{c1c2-bounds}. In comparing the relevant expressions, some useful relations are
\begin{equation}
\Delta_{2}(\hat{\eta})^{2}+\cc_{1}=\frac{9(x^{2}+c_{1})}{(x_{1}+x_{2}+x_{3})^{2}}\,,\qquad\kappa=\frac{3\ell\Delta_{0}}{(x_{1}+x_{2}+x_{3})^{2}}\,.
\end{equation}
Finally, note that the parameter region can be extended to $\cb^{2}>\ca^{2}$. To see this, observe that if we exchange
\begin{equation}\label{exchange}
\hat{\phi}^{1}\leftrightarrow\hat{\phi}^{2}\,,\qquad\hat{\eta}\leftrightarrow-\hat{\eta}\,,\qquad\mathcal{A}^{2}\leftrightarrow\mathcal{B}^{2}\,,
\end{equation}
we get identical expressions in subsection \ref{sec:NH} for the near-horizon solution.

\subsection{Enhanced symmetry}

The near-horizon solution with enhanced local $SU(2)\times U(1)$ rotational symmetry is parametrised by a constant $\Delta>0$ and the values of the  constant scalars $X^{I}>0$ subject to $X^{1}X^{2}X^{3}=1$~\cite{Kunduri:2007qy}. The data in \eqref{NH-metric-main} and \eqref{NH-gauge-main} explicitly read
\begin{align}
\Delta^{(0)}= & \Delta\,,\nonumber \\
h^{(0)}= & \frac{X\Delta}{\ell(\Delta^{2}+\ell^{-2}Y)}(\td\phi+\cos\theta\td\psi)\,,\nonumber \\
\gamma^{(0)}= & \frac{1}{\Delta^{2}+\ell^{-2}Y}(\td\theta^{2}+\sin^{2}\theta\td\psi^{2})+\frac{\Delta^{2}}{(\Delta^{2}+\ell^{-2}Y)^{2}}(\td\phi+\cos\theta\td\psi)^{2}\,,
\end{align}
and
\begin{equation}
a^{I}_{(0)}=-X^{I}\frac{(X-2X^I)}{\ell(\Delta^{2}+\ell^{-2}Y)}(\td\phi+\cos\theta\td\psi)\,,
\end{equation}
where
\begin{equation}
X=X^{1}+X^{2}+X^{3}\,,
\end{equation}
\begin{equation}
Y=(X^{1})^{2}+(X^{2})^{2}+(X^{3})^{2}-2X^{1}X^{2}-2X^{1}X^{3}-2X^{2}X^{3}\,,
\end{equation}
and the ranges of Euler angles that cover $S^3$ are $0\leq \phi \leq4\pi , 0\leq \psi \leq 2\pi, 0\leq \theta\leq\pi$.

The above data is equivalent to  \eqref{NH-expl1} and \eqref{NH-expl2} for $\cb^{2}=\ca^{2}$  (up to gauge transformations) under the coordinate change
\begin{equation}
\phi=\hat{\phi}^{1}+\hat{\phi}^{2}\,,\qquad\psi=-\hat{\phi}^{1}+\hat{\phi}^{2}\,,\qquad\cos\theta=\hat{\eta}\,,
\end{equation}
and the parameters are related by 
\begin{equation}
\Delta=\frac{3\kappa}{\ell\hat{H}^{2/3}}\,,\qquad X^{I}=\frac{\hat{H}^{1/3}}{1-\ca^{2}+\ck_{I}}\,,
\end{equation}
where the function $\hat{H}(\hat{\eta})$ for $\cb^{2}=\ca^{2}$ becomes a constant. The inverse transformation of the parameters is
\begin{equation}
\ca^{2}=\frac{\ell^{2}\Delta^{2}+Y}{\ell^{2}\Delta^{2}+X^{2}}\,,\qquad\ck_{I}=\frac{X^{2}-Y}{\ell^{2}\Delta^{2}+X^{2}}\frac{2-(X^{I})^{2}(X-X^{I})}{1+(X^{I})^{2}(X-X^{I})}\,.
\end{equation}
Note that $X^{1}X^{2}X^{3}=1$ implies \eqref{ck-constraint} and that the constraints $\Delta>0$ and $X^{I}>0$ become \eqref{kappa2} and \eqref{c1c2-bounds} respectively.

\end{document}